\documentclass[a4paper, 12pt]{article}

\usepackage[sort&compress]{natbib}
\bibpunct{(}{)}{;}{a}{}{,} 

\usepackage{amsthm, amsmath, amssymb, mathrsfs, multirow, url}
\usepackage[caption = false]{subfig}
\usepackage{graphicx} 
\usepackage{ifthen} 
\usepackage{amsfonts}
\usepackage[usenames]{color}
\usepackage{fullpage}

\theoremstyle{plain} 
\newtheorem{thm}{Theorem}

\theoremstyle{definition}

\theoremstyle{remark} 
\newtheorem*{astep}{A-step}
\newtheorem*{pstep}{P-step}
\newtheorem*{cstep}{C-step}

\newcommand{\prob}{\mathsf{P}}

\newcommand{\pl}{\mathsf{pl}}
\newcommand{\mpl}{\mathsf{mpl}}

\newcommand{\unif}{{\sf Unif}}
\newcommand{\nm}{{\sf N}}

\newcommand{\chisq}{{\sf ChiSq}}

\newcommand{\RR}{\mathbb{R}}

\newcommand{\YY}{\mathbb{Y}}
\newcommand{\UU}{\mathbb{U}}

\renewcommand{\S}{\mathcal{S}}

\newcommand{\iid}{\overset{\text{\tiny iid}}{\,\sim\,}}
\newcommand{\ind}{\overset{\text{\tiny ind}}{\,\sim\,}}

\usepackage[linesnumbered,algoruled,boxed,lined,commentsnumbered]{algorithm2e} 

\title{Generalized inferential models for meta-analyses based on few studies}
\author{Joyce Cahoon\footnote{Department of Statistics, North Carolina State University; {\tt jyu21@ncsu.edu}, {\tt rgmarti3@ncsu.edu}} \quad \text{and} \quad Ryan Martin$^*$}
\date{\today}

\begin{document}

\maketitle 

\begin{abstract}   
Meta-analysis based on only a few studies remains a challenging problem, as an accurate estimate of the between-study variance is apparently needed, but hard to attain, within this setting.  Here we offer a new approach, based on the {\em generalized inferential model} framework, whose success lays in marginalizing out the between-study variance, so that an accurate estimate is not essential.  We show theoretically that the proposed solution is at least approximately valid, with numerical results suggesting it is, in fact, nearly exact.  We also demonstrate that the proposed solution outperforms existing methods across a wide range of scenarios.  
\smallskip

\end{abstract}

\section{Introduction}
\label{S:intro}  

The most important scientific questions are likely to be pursued by multiple researchers, resulting in separate analyses that, when appropriate, can be combined via a single {\em meta-analysis} to attain stronger and more definitive conclusions.  But even when it is appropriate to combine the results from multiple studies, there is often a non-negligible amount of between-study heterogeneity, which is difficult to estimate accurately when the number of studies for meta-analysis is small.  Interestingly, meta-analyses with as few as three studies are the norm, not the exception \citep{davey2011}, so there is considerable interest in developing improved methods for inference in this setting of combining results from just a few heterogeneous studies. 


To set the scene, consider the classical normal--normal random-effects model where each study $k$ included for meta-analysis provides data $(Y_k, \sigma_k^2)$. These are modeled as 
\begin{equation}
(Y_k \mid M_k) \ind \nm(M_k, \sigma_k^2), \quad M_k \iid \nm(\mu, \nu), \quad k=1,\ldots,K.
\label{eq:generative}
\end{equation}
Here, $M_k$ denotes the random effect from study $k$, $\sigma_k^2 > 0$ the variance within study $k$, $\mu$ the underlying population effect and $\nu > 0$ the between-study variance.  This hierarchical formulation of the model indicates that the $K$ studies have something in common, namely, a tendency toward $\mu$, but that they are not fully homogeneous and that this degree of heterogeneity, controlled by $\nu$, is unknown.  Therefore, the goal is inference on the unknown mean $\mu$, with $\nu$ as an unknown nuisance parameter.  

In the simple case where $\nu$ is {\em known}, the meta-analysis is straightforward.  Marginally, $Y_k$ are independent, distributed as $\nm(\mu, \sigma_k^2 + \nu)$, for $k=1,\ldots,K$.  It is easy to check that the minimum variance unbiased estimator of $\mu$ is 
\[ \hat\mu(\nu) = \frac{\sum_{k=1}^K w_k(\nu) \, Y_k}{\sum_{k=1}^K w_k(\nu)}, \quad \text{where} \quad w_k(\nu) = (\sigma_k^2 + \nu)^{-1}, \]
and its variance is $\{\sum_{k=1}^K w_k(\nu)\}^{-1}$.  Of course, the more common scenario in applications is that $\nu$ is {\em unknown} and, somehow, the data must be used to account for that additional uncertainty.  A first idea that comes to mind is to estimate $\nu$ by some function $\hat\nu$ of the data, and then plug that into the formula for the estimator of $\mu$, i.e., $\hat\mu(\hat\nu)$.  The most well-known strategy is that proposed by \citet{dersimonian1986}, but there are others, e.g., \citet{paule1982}, \cite{cochran1954}.  Unfortunately, the number of studies, $K$, to be combined is often relatively small \citep{davey2011}, say, $K \leq 7$, and obtaining reliable estimates of the between-study variance based on so few samples is a challenge.  Besides, this plug-in style procedure does not naturally account for uncertainty in $\nu$, so any inference drawn can only be (provably) valid in an asymptotic ($K \to \infty$) sense, which may not be especially meaningful in applications where $K \leq 7$.  More tangibly, confidence intervals based on the DerSimonian--Laird plug-in style method have shown to perform poorly even when $K$ is as large as 20 \citep{liu2018}; see, also, \citet{viechtbauer2005}, \citet{dersimonian2007}, \citet{sidik2007}, \citet{jackson2010}, \citet{chung2013}, and \citet{veroniki2016}.  Therefore, there is a desire for alternative methods that marginalize out the nuisance parameter, $\nu$, and achieve frequentist performance even when the number of studies is small.   

To alleviate this problem of estimating $\nu$, likelihood methods have been proposed, e.g., the log likelihood ratio in \citet{goodman1989meta}, the profile likelihood in \citet{hardy1996}, and the signed profile log likelihood in \citet{severini2000likelihood}, among others.  Unlike previous plug-in methods, these likelihood based procedures introduced an appropriate widening in confidence intervals to deal with the imprecision in estimating $\nu$.  In fact, \citet{guolo2012higher} was able to improve upon these first-order inference results by introducing a Skovgaard correction to the signed profile log likelihood, making it asymptotically standard normal with an error of $O(n^{-3/2})$.  Resorting to higher-order asymptotics with the Skovgaard correction, however, still failed to modify the likelihood in the manner necessary to achieve nominal coverage when the number of studies available is small. 
 
So-called {\em exact methods}, featuring such frequentist guarantees, have been offered as alternatives to the plug-in style and likelihood methods described above.  These include the methods in \citet{follmann1999}, \citet{liu2018}, and \citet{wang2018} which, in one way or another, are based on permutation distributions.  While these permutation methods can produce confidence intervals that achieve nominal frequentist coverage, the discreteness of the permutation distribution makes the results overly conservative, unless $K$ is relatively large.  To our knowledge, the most recent entry into the literature on exact methods for meta-analysis is \citet{michael2018}, and since this shares a number of similarities to our proposal, at least in terms of its construction, we describe this in some detail in Section~\ref{SS:gim.construct}.  

Given that the goal is to develop a method for meta-analysis that controls Type I error, even when $K$ is small, it makes sense to consider the one general, normative framework we are aware of that offers such guarantees.  Specifically, \citet{martin2013} present a construction of what they call an {\em inferential model} that leads to provably valid inference, no asymptotic justification required; see \citet{imbook} for a monograph-length introduction, and \cite{martin2019false} for a survey of some recent developments.  The distinguishing feature of this approach is the user-specified random set that leads to a sort of  ``posterior'' (but {\em not} a probability measure) on the parameter space.  The distribution of this random set can then be used to visualize the information data provides about the parameter of interest as well as construct inference procedures.  We will briefly review the basic construction and its properties in Section~\ref{SS:bim}.  It may happen that the basic inferential model construction is difficult to carry out in an application, so \citet{martin2015, martin2018} developed a simpler and more direct {\em generalized} version.  Since this is the strategy we follow here for meta-analysis, we provide a brief review in Section~\ref{SS:gim}.  

This meta-analysis application boils down to a marginal inference problem, i.e., $\nu$ is an unknown nuisance parameter to be marginalized out so that we can make inference about $\mu$.  The generalized inferential model framework provides at least two strategies for marginalization, and we will show that the method proposed in \citet{michael2018} is itself a generalized inferential model based on one particular choice of marginalization strategy.  Our proposed method, on the other hand, is based on a different and arguably more natural choice of marginalization strategy, leading to a method that performs better than theirs in a variety of respects.  In Section~\ref{S:immeta}, we describe the construction of a generalized inferential model for meta-analysis, show that the solution in \citet{michael2018} is a special case, and present our proposed method.  Details about the computation and theoretical justification are also provided.  Numerical comparisons in simulated data experiments demonstrate that our proposed method outperforms existing methods in terms of both validity and efficiency across a broad range of scenarios.  Two real applications are presented in Section~\ref{S:example} and some concluding remarks are given in Section~\ref{S:discuss}.

\section{Background}
\label{S:im}

\subsection{Basic inferential models}
\label{SS:bim}

Fisher and later Dempster aimed to develop a framework of probabilistic inference without prior distributions, i.e., a prior-free alternative to Bayesian inference.  These approaches, however, failed to reach the statistical mainstream, largely because the derived procedures have no frequentist guarantees.  To fill that gap, \citet{martin2013} argued that frequentist guarantees could be achieved by supplementing the structural, pivotal, or functional model formulation of \citet{fisher1956}, \citet{dempster2008}, \citet{fraser1968}, \citet{barnard1995pivotal}, \citet{dawid1982functional}, \citet{taraldsen2013fiducial},
and others, with an appropriate user-specified random set.  As a consequence of this use of a random set, the inferential output is described by a (data-dependent) non-additive plausibility function \citep[e.g.,][]{shafer1976} instead of an additive posterior probability distribution.  The added complexity of non-additivity is not for its own sake, however, it is actually necessary \citep{balch2019} for the strong validity property in \eqref{eq:valid} that leads to frequentist error rate control.  

In general, we start with a sampling model, $Y \sim \prob_{Y|\theta}$, for the observable data $Y$, depending on some unknown parameter $\theta \in \Theta$.  As one would if the goal were to simulate, we express the model as 
\begin{equation}
\label{eq:baseline.assoc}
Y = a(\theta, U), \quad U \sim \prob_U. 
\end{equation}
where $a$ is a known function, and $\prob_U$ is a known distribution on the space $\UU$.  
We call the relationship in \eqref{eq:baseline.assoc} an {\em association} between observable data $Y$, unknown parameter $\theta$, and auxiliary variable $U$.  This association step is the starting point in the construction of a valid inferential model. 

\begin{astep}
Specify an association of the form \eqref{eq:baseline.assoc} and then define the set 
\[ \Theta_y(u) = \{\theta: y = a(\theta,u)\}, \quad u \in \UU. \]
\end{astep}

\begin{pstep}
Specify a predictive random set $\S$, taking values in the power set of $\UU$, whose contour function, $f(u) = \prob_\S(\S \ni u)$, is such that $f(U) \sim \unif(0,1)$, when $U \sim \prob_U$.
\end{pstep}

\begin{cstep}
Combine the ingredients in the A- and P-steps to get a new random set 
\[ \Theta_y(\S) = \bigcup_{u \in \S} \Theta_y(u), \]
and, for inference about $\theta$, return the distribution of this random set summarized by its plausibility function
\[ \pl_y(A) = \prob_\S\{\Theta_y(\S) \cap A \neq \varnothing\}, \quad A \subseteq \Theta. \]
The distribution of this random set is interpreted as a measure of how {\em plausible} the hypothesis ``$\theta \in A$'' is, based on data $y$ and the posited model.
\end{cstep}

The most unique feature of this construction is the random set $\S$.  In our meta-analysis problem, specification of the random set is straightforward as outlined in Section~\ref{SS:gim}, but the general details can also be found in \citet{imbook}.  What matters is that the properties required of the random set make the inferential model {\em valid}, i.e., 
\begin{equation}
\label{eq:valid}
\sup_{\theta \in A} \prob_{Y|\theta}\{ \pl_Y(A) \leq \alpha\} \leq \alpha \quad \text{for all $\alpha \in (0,1)$ and all $A \subseteq \Theta$}. 
\end{equation}
An important consequence of this validity property is the control it provides on the performance of statistical procedures derived from the inferential model output.  Indeed, a test that rejects a hypothesis ``$\theta \in A$'' if $\pl_y(A) \leq \alpha$ will obviously control the frequentist Type~I error rate at level $\alpha$.  Similarly, a $100(1-\alpha)\%$ plausibility region for $\theta$, given by 
\[ \{\theta: \pl_y(\theta) > \alpha\}, \quad \text{where $\pl_y(\theta) := \pl_y(\{\theta\})$}, \]
has frequentist coverage probability of (at least) $1-\alpha$.  These properties are exact in the sense that they do not require any asymptotic approximations.  The pointwise plausibility function, $\theta \mapsto \pl_y(\theta)$, is also a useful visualization tool, not unlike a Bayesian posterior density function; see Figure~\ref{fig:imexample} below.

\subsection{Generalized inferential models}
\label{SS:gim}

Since the inferential model construction, and corresponding validity result, is general, efficiency often becomes a concern as the dimension of the auxiliary domain grows with the dimension of the data.  To avoid the possible need to specify such a complex, high-dimensional random set, some non-trivial manipulations are required as described \citep[e.g.,][]{martin2015marginal, martin2015conditional} that can be difficult to carry out in a given problem.  
This motivated \citet{martin2015, martin2018} to develop a construction based on a more general formulation and establish conditions under which the corresponding inferential model is valid.  An advantage of this generalized approach is that there is no need for the aforementioned manipulations, hence it is easier to apply.  

A generalized inferential model begins with defining a real-valued function $(y,\theta) \mapsto T_y(\theta)$.  When $Y \sim \prob_{Y|\theta}$, the random variable $T_Y(\theta)$ has a distribution, which we represent with $G_\theta$.  The generalized association then extends the notion in Section~\ref{SS:bim} by connecting the data, parameter, and auxiliary variable via the expression
\[ T_Y(\theta) = G_\theta^{-1}(U), \quad U \sim \unif(0,1). \]
Here and throughout this paper, we will assume (without loss of generality) that $T_y(\theta)$ is large when data $y$ and parameter value $\theta$ disagree.  The first step to our generalized inferential model (A-step) then yields 
\[ \Theta_y(u) = \bigl\{\theta: G_\theta\bigl( T_y(\theta) \bigr) = u \bigr\}, \quad u \in (0,1). \]
The P-step, as before, requires the introduction of some random set in the $u$-space; but the structure that has been imposed here virtually determines it. We thus take 
\[ \S = (0, \tilde U), \quad \tilde U \sim \unif(0,1). \]
The C-step returns a new random set 
\[ \Theta_y(\S) = \bigcup_{u \in \S} \Theta_y(u) = \bigl\{ \theta: G_\theta\bigl( T_y(\theta) \bigr) \leq \tilde U \bigr\}, \quad \tilde U \sim \unif(0,1). \]
Note that this set contains those parameter values that agree with $y$ to some degree, and that this degree is calibrated so that validity holds.  That is, if 
\[ \pl_y(\theta) = \prob_\S\{\Theta_y(\S) \ni \theta\} = \prob_{\tilde U}\{G_\theta(T_Y(\theta)) \leq \tilde U\} = 1 - G_\theta(T_y(\theta)), \]
then we immediately see that $\pl_Y(\theta) \sim \unif(0,1)$ under $Y \sim \prob_{Y|\theta}$ and, hence, the validity property as stated in \eqref{eq:valid} holds.  

The approach described above returns a plausibility function defined on the full parameter space, $\Theta$.  From this, one can carry out marginal inference on any feature, $\psi = \psi(\theta)$, of $\theta$ via optimization.  In particular, following \citet{shafer1987} Sec.~G, the corresponding marginal point plausibility function for $\psi$ is 
\begin{equation}
\label{eq:mpl0}
\mpl_y(\psi) = \sup_{\theta: \psi(\theta) = \psi} \pl_y(\theta), 
\end{equation}
and the validity properties associated with $\pl$ carry over immediately to $\mpl$.  

But there are cases, like the one we consider in Section~\ref{S:immeta}, where interest is exclusively in a specific feature of $\theta$, and it is beneficial to construct a marginal inferential model directly.  Express the full parameter as $\theta = (\psi, \eta)$, where $\psi$ and $\eta$ are the interest and nuisance parameters, respectively.  Next, define a function $(y,\psi) \mapsto T_y(\psi)$ that only directly involves the interest parameter, again, with the property that large values of the function correspond to cases where data and the interest parameter disagree.  An example of such a function is the negative log relative profile likelihood, as in \eqref{eq:relprof}.  If the distribution of $T_Y(\psi)$, as a function of $Y \sim \prob_{Y|\psi,\eta}$, does not depend on $\eta$, then construction of the generalized inferential model for $\psi$ proceeds exactly as above, with any fixed value of $\eta$.  In most applications, however, including our meta-analysis problem, the distribution of $T_Y(\psi)$ does depend on the nuisance parameter, so some non-trivial adjustments are required.  We discuss this in detail in Section~\ref{S:immeta}.

\section{Inferential models for meta-analysis}
\label{S:immeta}  

\subsection{Construction}
\label{SS:gim.construct}

For our meta-analysis case, write $\prob_{Y|\mu,\nu}$ for the joint distribution of $Y = (Y_1,\ldots,Y_K)$, where $Y_k$ are independently generated from a $\nm(\mu, \sigma_k^2 + \nu)$. It is straightforward to write down an association that links the data $Y$, the unknown parameter $\theta=(\mu,\nu)$, and a set of auxiliary variables, e.g., 
\[ Y_k = \mu + (\sigma_k^2 + \nu)^{1/2} U_k, \quad k=1,\ldots,K, \]
where $U_k$'s are iid $\nm(0,1)$. Following \citet{martin2015conditional}, the next step would be to reduce the dimension of $(U_1,\ldots,U_K)$ to match that of $\theta$.  This step turns out to be challenging but, fortunately, a generalized inferential model is within reach.  

The full parameter is $\theta=(\mu,\nu)$ but, since only $\mu$ is of interest, marginalization is desired.  As discussed in Section~\ref{SS:gim}, there are at least two ways to proceed.  The first is to start with a summary $T_Y(\theta) = T_Y(\mu,\nu)$ of the data and full parameter, which takes large values when data and the candidate parameters disagree, and then marginalize to the $\mu$-space {\em after} constructing the plausibility function on the $(\mu,\nu)$-space.  That is, we define a plausibility function on the full parameter space as
\[ \pl_y(\mu,\nu) = 1 - G_{\mu,\nu}\bigl( T_y(\mu,\nu) \bigr), \]
where $G_{\mu,\nu}$ is the distribution of $T_Y(\mu,\nu)$ under $Y \sim \prob_{Y|\mu,\nu}$.  Like in \eqref{eq:mpl0}, we obtain our desired marginal plausibility by optimization:
\[ \mpl_y(\mu) = \sup_\nu \pl_y(\mu, \nu). \]
The corresponding $100(1-\alpha)$\% plausibility interval for $\mu$ is $\{\mu: \mpl_y(\mu) > \alpha\}$, corresponding to a projection of the joint plausibility region for the full parameter onto the $\mu$-space.  After some reflection on the solution in \citet{michael2018}, one sees that it is precisely a generalized inferential model as just described, with $T_y(\mu,\nu)$ given by 
\begin{equation}
T_y(\mu, \nu) = T_w(\mu) + c_0 \Big[ \log \frac{L_y(\mu, \hat{\nu}_{DL})}{L_y(\mu, \nu)}
\Big], 
\label{eq:exacteq}
\end{equation}
a linear combination of DerSimonian and Laird's Wald-type summary statistic $T_w(\mu)$ and a log likelihood ratio, with a constant $c_0$ controlling its contribution.  
Those authors do not describe their proposal as a (generalized) inferential model, but we believe that this perspective is beneficial both for developing some intuition about their solution and for comparing with our proposed solution.  

Despite the ease of marginalizing out the nuisance parameter from the joint plausibility function, we adopt the second strategy for marginalization discussed in Section~\ref{SS:gim} that eliminates the nuisance parameter {\em before} constructing the plausibility function.  That is, we start with a function, $T_y(\mu)$, that does not directly involve $\nu$.  Like in \citet{goodman1989meta}, we recommend the use of a negative log relative profile in which
\begin{equation}
T_y(\mu) = -\log \frac{\sup_\nu L_y(\mu, \nu)}{\sup_{\mu,\nu} L_y(\mu, \nu)}, \quad \mu \in \RR. 
\label{eq:relprof}
\end{equation}
Here $L_y(\mu, \nu) \propto \prod_{k=1}^K (\sigma_k^2 + \nu)^{-1/2} \exp\{-\frac12 (\sigma_k^2 + \nu)^{-1} (y_k - \mu)^2\}$ is the likelihood function under the assumed model $\prob_{Y|\mu,\nu}$.  Note that the $\nu$ value at which the maximum is attained in the numerator---call it $\hat\nu_\mu$---depends on the specified value of $\mu$.  As before, we define the distribution function of $T_Y(\mu)$ under the model $Y \sim \prob_{Y|\mu,\nu}$ as 
\[ G_\nu(t) = \prob_{Y|0,\nu}\{T_Y(0) \leq t\}, \quad t > 0. \]
Here we have inserted the default zero value for $\mu$ because the location model structure means the distribution of $T_Y(\mu)$ does not depend on the value of $\mu$, when $Y \sim \prob_{Y|\mu,\nu}$.  Following the remainder of the construction outlined in Section~\ref{SS:gim}, we arrive at a marginal point plausibility function 
\[ \text{``$\mpl_y(\mu)$''} = 1 - G_\nu(T_y(\mu)). \]
The quotation marks on the left-hand side are to signal that this is not a function that we can actually work with because the right-hand side depends on the unknown value of the nuisance parameter $\nu$.  To overcome this, we will use a plug-in estimate for $\nu$, where it appears in $G_\nu$.  Before proceeding, it is important to emphasize that our plug-in proposal is fundamentally different than those mentioned in Section~\ref{S:intro}; we discuss this in more detail in Section~\ref{SS:valid}.  For our plug-in estimator, we propose to use that $\nu$ value where the maximum in the numerator of the profile likelihood is attained, namely, $\hat\nu_\mu$, which implicitly depends on data $y$.  Putting it all together, our proposed marginal point plausibility function for $\mu$ is 
\begin{equation}
\label{eq:mpl}
\mpl_y(\mu) = 1-G_{\hat\nu_\mu}(T_y(\mu)), \quad \mu \in \RR, 
\end{equation}
which is now just a function of data and the generic argument $\mu$.  This function can be plotted to visualize what the data suggests about where the true value of $\mu$ is and, more formally, we can read off a $100(1-\alpha)\%$ marginal plausibility interval for $\mu$ as follows:
\[ \{\mu: \mpl_y(\mu) > \alpha\}. \]

Computation of $\mpl_y$ requires an approximation of the analytically intractable distribution $G_{\hat\nu_\mu}$, but this is straightforward to do via Monte Carlo; see Algorithm~\ref{algo:1}.  And once $\mpl_y(\mu)$ is available on a grid of values, extracting the plausibility interval for $\mu$ is easy, but the endpoints could be targeted more directly using, say, the proposed Monte Carlo method coupled with stochastic approximation. 

\begin{algorithm}[t]
 Generate $M$ samples of $K$ study-level errors $e_{1m}^*, \ldots, e_{Km}^* \sim \nm(0,1)$\; 
 Set a fine grid of $\mu$ values\;
 \For{each $\mu$ value on the specified grid} {
 Find $\hat\nu_{\mu}$ for the observed data $y$ and given $\mu$\;
 \For{$m=1,\ldots,M$}{
  Set $Y_{km}^* = (\sigma_k^2 + \hat\nu_{\mu})^{1/2}e_{km}^*$ for $k = 1, \ldots, K$\; 
  Calculate $T_m^* = T_{Y_m^*}(0)$ based on $Y_m^* = (Y_{1m}^*,\ldots,Y_{Km}^*)$\; 
 }
 Approximate $G_{\hat\nu_\mu}(t)$ by $M^{-1}\sum_{m=1}^M 1\{T_m^* \leq t\}$\; 
 }
 \caption{Monte Carlo approximation of $G_{\hat\nu_\mu}$}
 \label{algo:1}
\end{algorithm}

\subsection{Illustration}

To illustrate our proposed method, two examples are shown in Figure~\ref{fig:imexample}.  In each, a meta-analysis is carried out on $K=3$ studies, and each study's variance $\sigma_k^2$ generated from a inverse gamma distribution with a shape and scale parameter of 1.  The data supplied from each of these hypothetical studies were generated from a normal distribution in which the true population mean was set at $\mu = 5$ and the variance at $\nu + \sigma_k^2$.  It is straightforward to construct plausibility functions for the individual studies, with data $(Y_k, \sigma_k^2)$, for $k=1,2,3$, under the normality assumption \citep[see, e.g.,][]{pvalue.course}, and these curves are plotted in gray in Figure~\ref{fig:imexample}.  The black curve corresponds to the marginal plausibility function, $\mu \mapsto \mpl_y(\mu)$, described in the previous section, evaluated using the Monte Carlo method outlined in Algorithm~\ref{algo:1}.  Panels~(a) and (b) correspond to $\nu=1$ and $\nu=2$, respectively.  Note that, as expected, the black curve is a ``combination'' of the three gray curves, with more influence coming from those gray curves that are tighter, corresponding to a more informative individual study.  

\begin{figure}%
    \centering
    \subfloat[$\nu=1$]{\includegraphics[width=7cm]{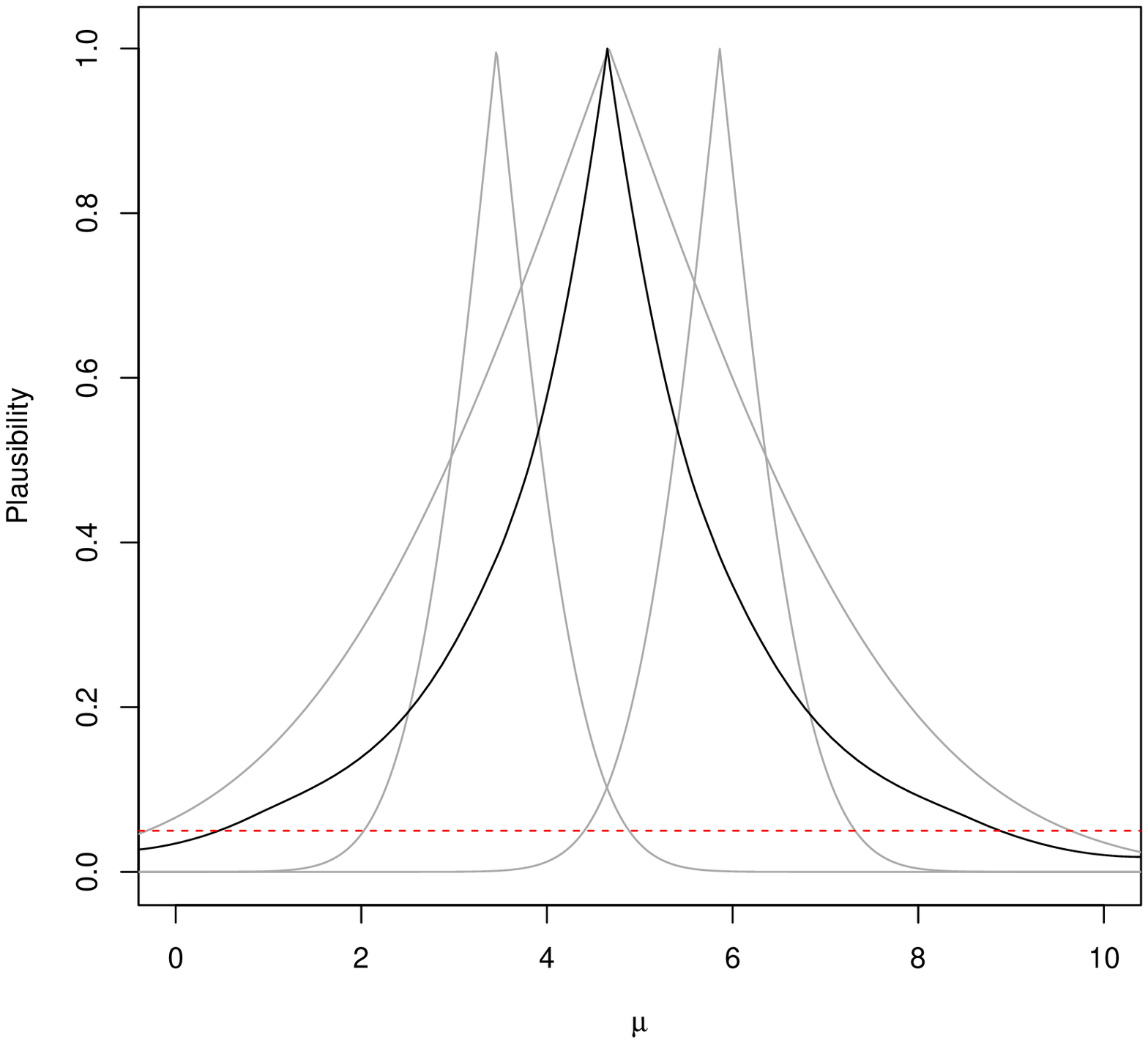}}%
    \qquad
    \subfloat[$\nu=2$]{\includegraphics[width=7cm]{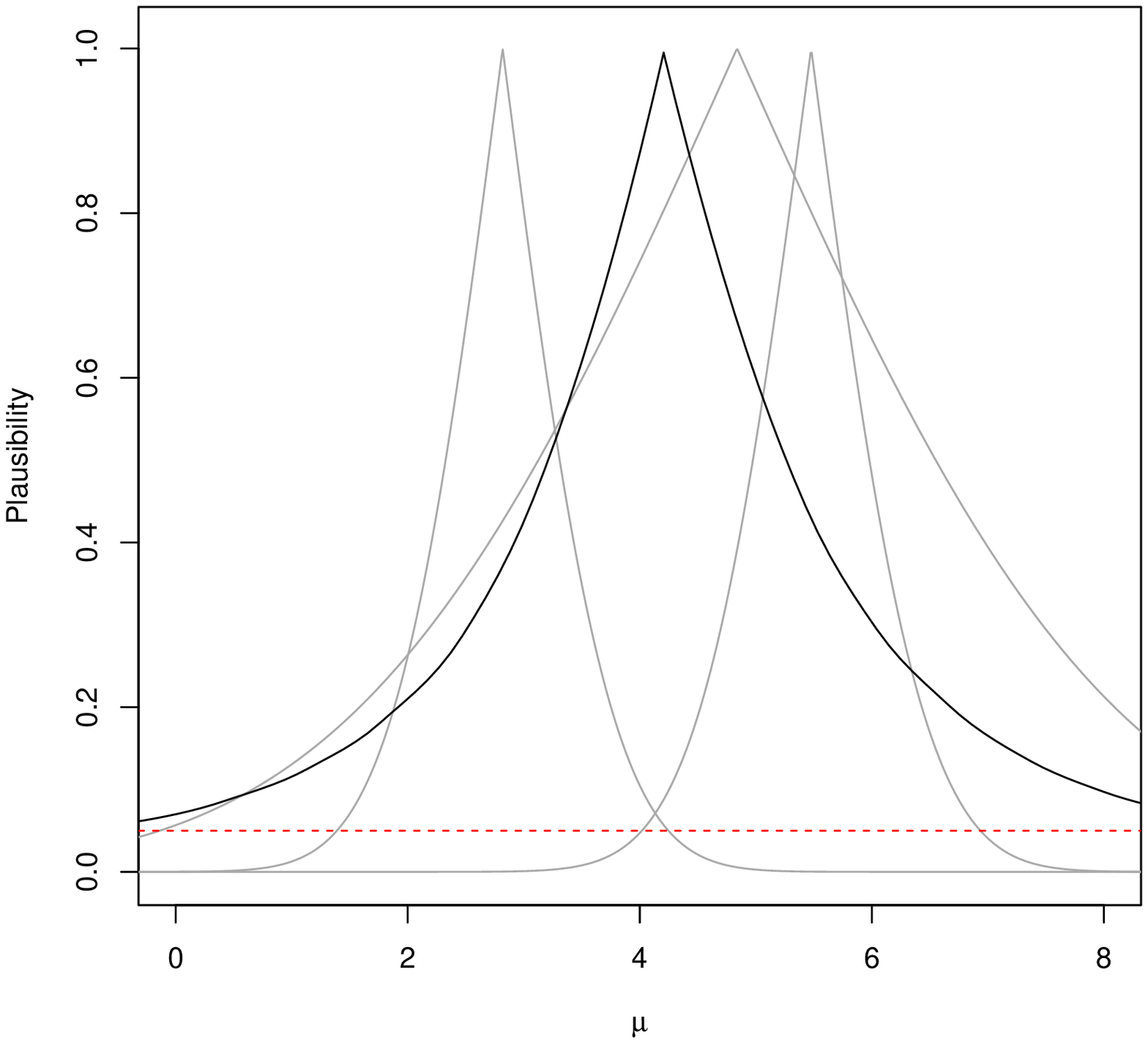} }%
    \caption{Examples of two simulated meta-analyses where the number of studies available $K = 3$. Plausibility functions associated with each individual study (in gray) and the combined plausibility function associated with our proposed inferential methods approach (in black). Marginal 95\% plausibility intervals for $\mu$ can be obtained where the combined plausibility intersects with $\alpha = 0.05$.}%
    \label{fig:imexample}%
\end{figure}

\subsection{Theoretical properties}
\label{SS:valid} 

If the value of the nuisance parameter $\nu$ were known, and used, in our construction of the (marginal) generalized inferential model for $\mu$, then the validity property, as stated in \eqref{eq:valid}, would be immediate.  For the practical case where $\nu$ is unknown, we have recommended the inferential model with marginal point plausibility function \eqref{eq:mpl}, which involves a plug-in estimator.  Our use, however, of this plug-in $\hat\nu_\mu$, complicates verification of the validity property.  At the very least, under mild assumptions, our proposed generalized inferential model would be valid for large $K$, and the following theorem confirms this.  

\begin{thm}
\label{thm:valid}
Let $Y^K = (Y_1,\ldots,Y_K)$ be an independent sample from the random effects model, $\prob_{Y|\mu,\nu}$, described above, where both $\mu$ and $\nu$ are unknown, but each within-study variance $\sigma_k^2$ is known. Then the marginal plausibility function $\mpl_{Y^K}$ in \eqref{eq:mpl} satisfies 
\[ \mpl_{Y^K}(\mu) \to \unif(0,1) \quad \text{in distribution under $\prob_{Y|\mu,\nu}$ as $K \to \infty$}. \]
In particular, the marginal plausibility region $\{\mu: \mpl_{Y^K}(\mu) > \alpha\}$ has coverage probability approximately equal to $1-\alpha$, for large $K$.  
\end{thm}

\begin{proof} 
See Appendix~\ref{S:proof}.
\end{proof}

The above theorem only says that the proposed solution is approximately valid for large $K$, but we do have reason to believe that the support for the proposed solution is actually stronger than this theorem suggests.  Indeed, numerically, even for small $K$, the distribution of $\mpl_Y(\mu)$ is very close to uniform.  As shown in Figure~\ref{fig:simvalid}, based on 10,000 samples of the data pairs $(Y_k, \sigma_k^2)$ from a small number of studies $K = \{3, 4, 5\}$ and a high level of heterogeneity between studies $\nu = 5$, the distribution is close to uniform.  There is some deviation to the left of uniform when the number of studies included for meta-analysis is particularly small, $K=3$, but this is in the middle of the distribution, not in the lower tails (e.g., around 0.05) where we would naturally be interested.  Therefore, the method appears to be not only valid, but nearly exact.  

It is natural to ask: why does our proposed method achieve this apparent higher-order of accuracy?  At least intuitively, this can be answered by noticing that our proposed method has features in common with both the exact and higher-order asymptotically accurate methods described above.  That is, by starting with the relative profile likelihood $T_Y(\mu)$ in \eqref{eq:relprof}, we remove almost all of the dependence on the nuisance parameter; that is, by Wilks's theorem, the profile likelihood ratio has a known distribution---no nuisance parameter dependence---up to first order.  This means that the exact distribution, $G_\nu$, of our $T_Y(\mu)$ is roughly constant in $\nu$.  Therefore, even though there is some remaining dependence on $\nu$, which is why a plug-in estimator is needed, it is not necessary that it be an especially accurate estimate.  Ultimately, our final inferential model is built using the plug-in distribution $G_{\hat\nu_\mu}$, at each individual $\mu$ value, which is very close to the exact distribution.  It is this extra accuracy that leads to the superior practical performance in Figure~\ref{fig:simvalid} and Section~\ref{S:sims}, beyond what would be expected from the large-$K$ approximate validity result in Theorem~\ref{thm:valid}.

\begin{figure}
\centering
\begin{minipage}{.33\textwidth}
  \centering
  \includegraphics[width=1\linewidth]{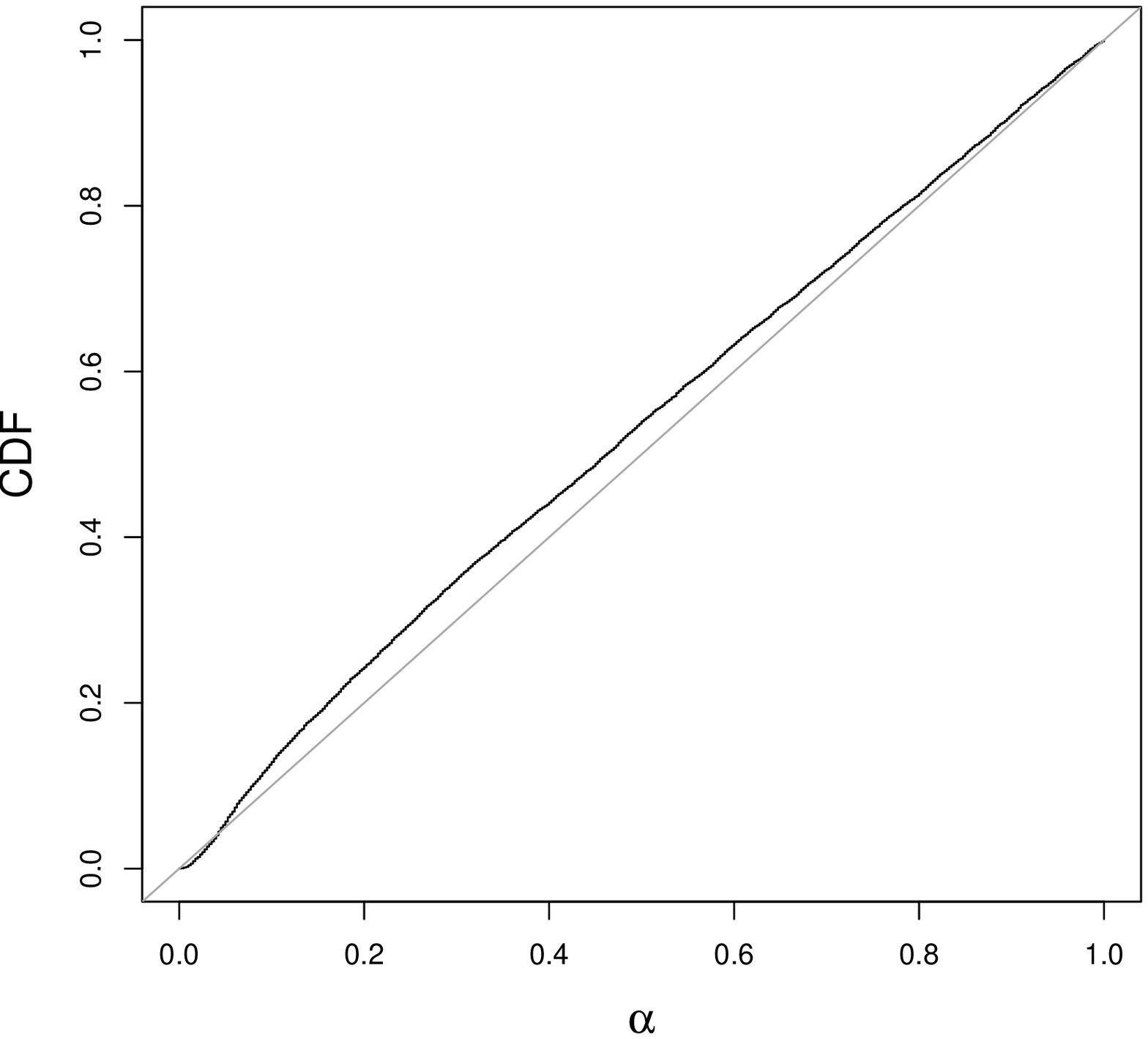}
\end{minipage}%
\begin{minipage}{.33\textwidth}
  \centering
  \includegraphics[width=1\linewidth]{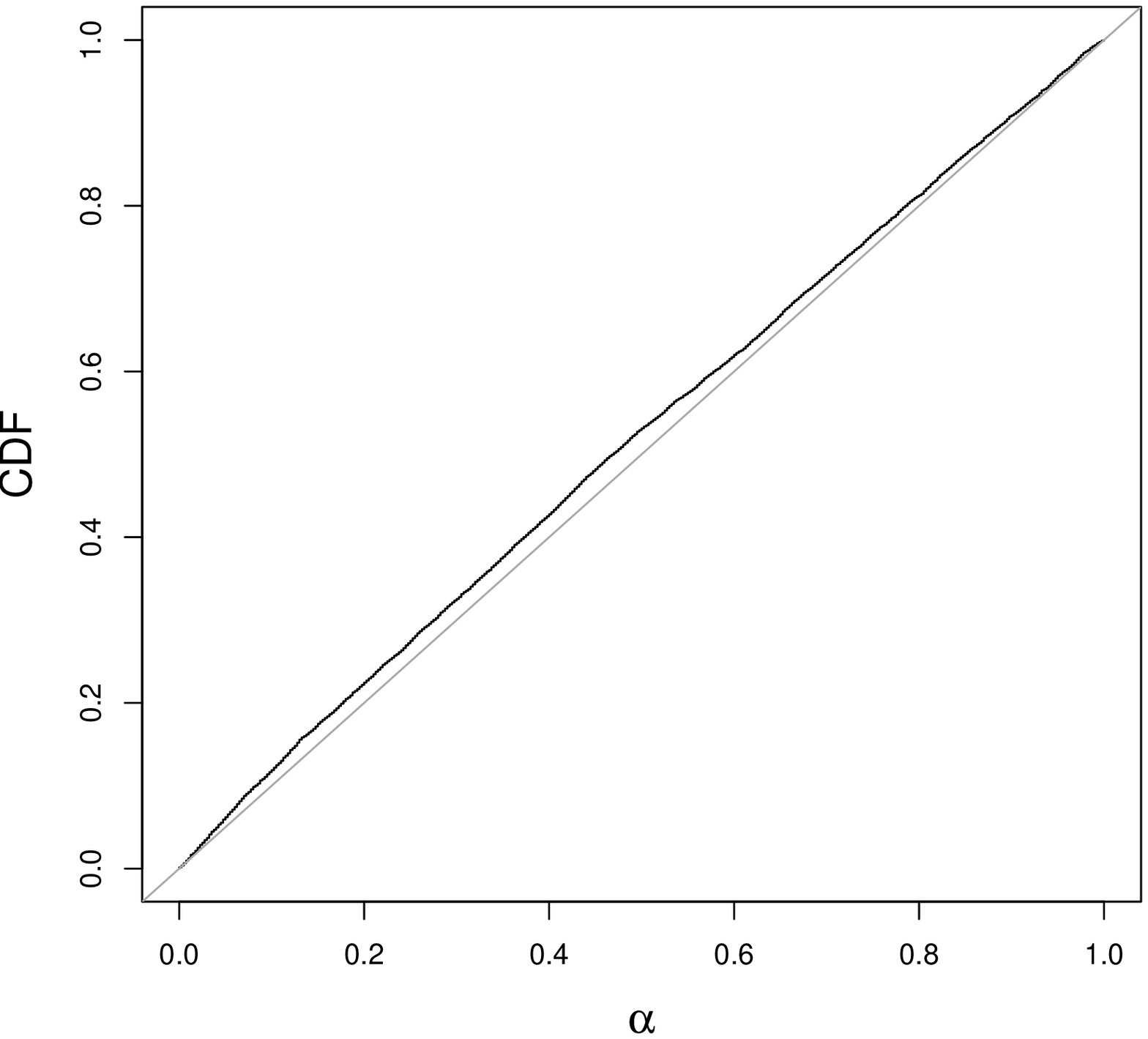}
\end{minipage}
\begin{minipage}{.33\textwidth}
  \centering
  \includegraphics[width=1\linewidth]{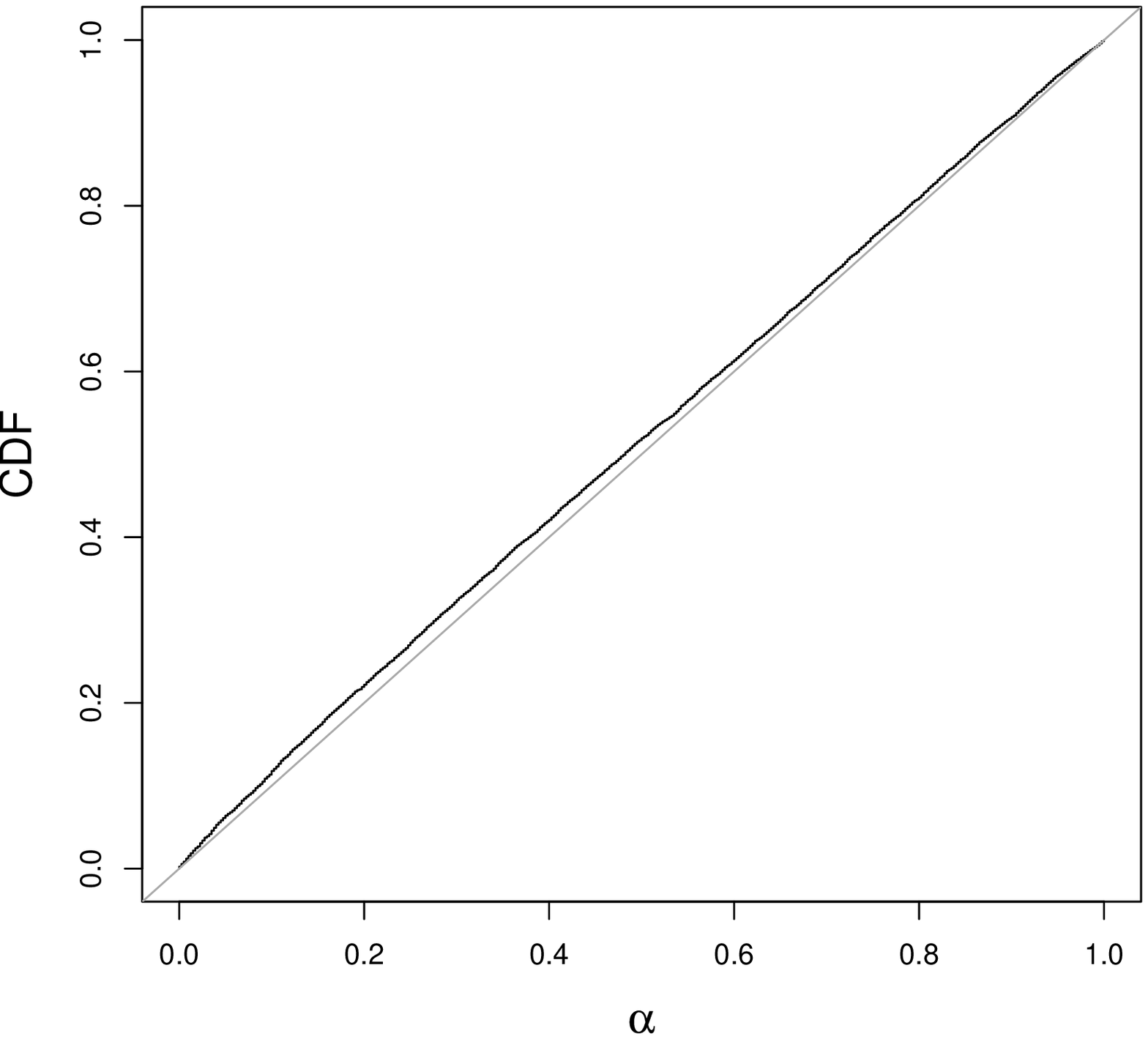}
\end{minipage}
\caption{Distribution of $\alpha \mapsto P_{Y| \mu, \nu} \{ \mpl_Y(\mu) \leq \alpha \}$ (in black) compared with that of a $\unif(0,1)$ (in gray) based on 10000 Monte Carlo samples simulated from a normal-normal random effects model in which $\mu = 5$, $\nu = 5$ and each study's variance is generated from an inverse gamma with a shape and scale parameter of 1. From left to right, the number of studies $K$ included is 3, 4 and 5.}%
\label{fig:simvalid}
\end{figure}


\section{Simulation studies}
\label{S:sims}  

Our simulations examine the performance of our proposed method, compared to that of other existing methods, in the case where $K \leq 7$.  We generate our $K$ study-level observations from the normal--normal random effects model, where the within-study variances---the $\sigma_k^2$'s---are taken as fixed constants; following \citet{gelman2006}, we generate these ``fixed values'' from an inverse gamma distribution with a shape and scale parameter of 1.  Here we fix the overall effect at $\mu=5$ but vary the between-study variances as $\nu \in \{1,3,5\}$, so that we capture various settings from low to high levels of heterogeneity.  For each combination of $K$ and $\nu$, we repeat the experiment 1000 times to get estimates of the coverage probability and mean length of various 95\% confidence intervals.  

We compare the coverage properties of our approach (IM) against that of \citet[DL,][]{dersimonian1986}, the exact method in \citet[EX,][]{michael2018}, the signed profile log likelihood ratio in \citet[LK,][]{severini2000likelihood}, its Skovgaard corrected cousin in \citet[SV,][]{guolo2012higher}, as implemented in the {\tt metaLik} package in R \citep{guolo2012r}, a traditional full Bayesian solution with a non-informative Jeffreys prior, as implemented in the {\tt bayesmeta} package \citep{rover2017}, and, as a benchmark, an {\em oracle} procedure that knows the true value of $\nu$ and uses the classical normal distribution theory for inference on $\mu$.  For the exact method, we set the tuning parameter $c_0$ to the values based on $K$ as recommended in \citet{michael2018}.  


As shown in Figure~\ref{fig:coverage_byk}, our proposed generalized inferential method outperforms all the other methods---except, of course, the oracle---in terms of both coverage and mean interval length.  
The Bayes, exact and higher-order likelihood methods tend to have too high nominal coverage, and the others too low.  The over-coverage seen here is consistent with the results shown in \citet{michael2018}.  Oddly, when the between-study variance parameter is set to a higher value $\nu = 5$, so that the average heterogeneity among the 1000 simulations is high, the Skovgaard corrected signed profile log likelihood actually achieves nominal coverage across all small settings of $K$.  This is in line with the results in \citet{guolo2012higher}, in which the estimator is sensitive to the level of heterogeneity.  

To further highlight this sensitivity of the higher-order likelihood approach, we re-ran our methods above in the same settings used in \citet{guolo2012higher}; more specifically, we re-examined the performance of our IM approach with $K \in \{3, 4, 5, 6, 7\}$, $\mu = 0.5$, and $\nu = \{0.08, 0.10, 0.12\}$.  We also generate the within-study variances for each study $K$ from a uniform distribution on the interval 0.01 and 0.06 as done in \citet{guolo2012higher}.  As shown in Figure~\ref{fig:coverage_new}, our method still outperforms in these settings.  Conversely, the under-coverage of DL and LK across these two different simulation settings in Figure~\ref{fig:coverage_byk} and Figure~\ref{fig:coverage_new} are to be expected, as $K$ is too small for the first-order asymptotic approximations to kick in.  Moreover, our proposed method's strong coverage performance is not the result of having overly wide intervals: our mean lengths fall right in between those of the over-and under-coverage methods, and are quite close to that of the oracle as $K$ becomes larger.  Remarkably, these patterns hold across different heterogeneity levels as well.

\begin{figure}[t]
   \centering
\begin{tabular}{ccc}
\includegraphics[width=.3\linewidth]{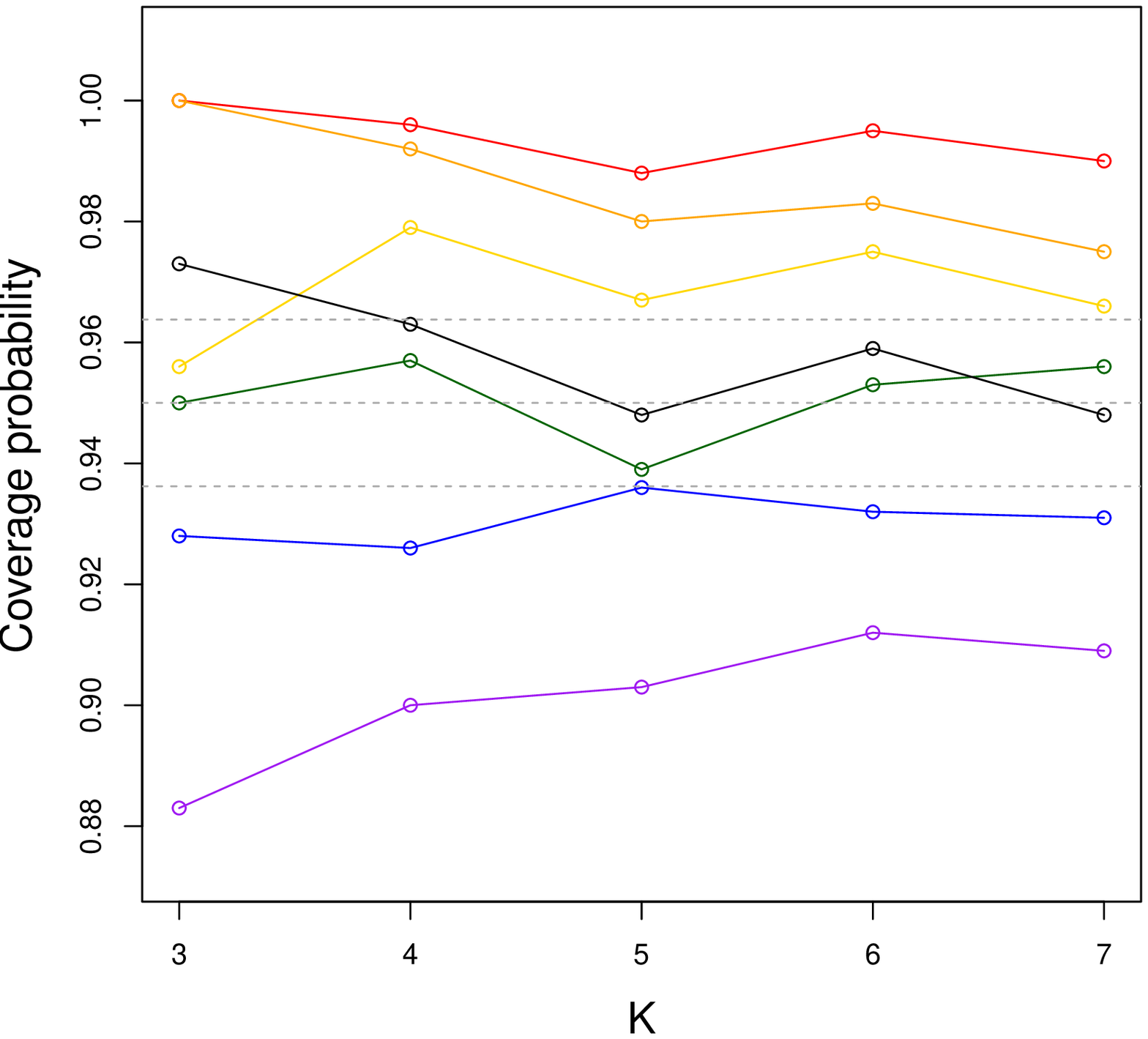}&
\includegraphics[width=.3\linewidth]{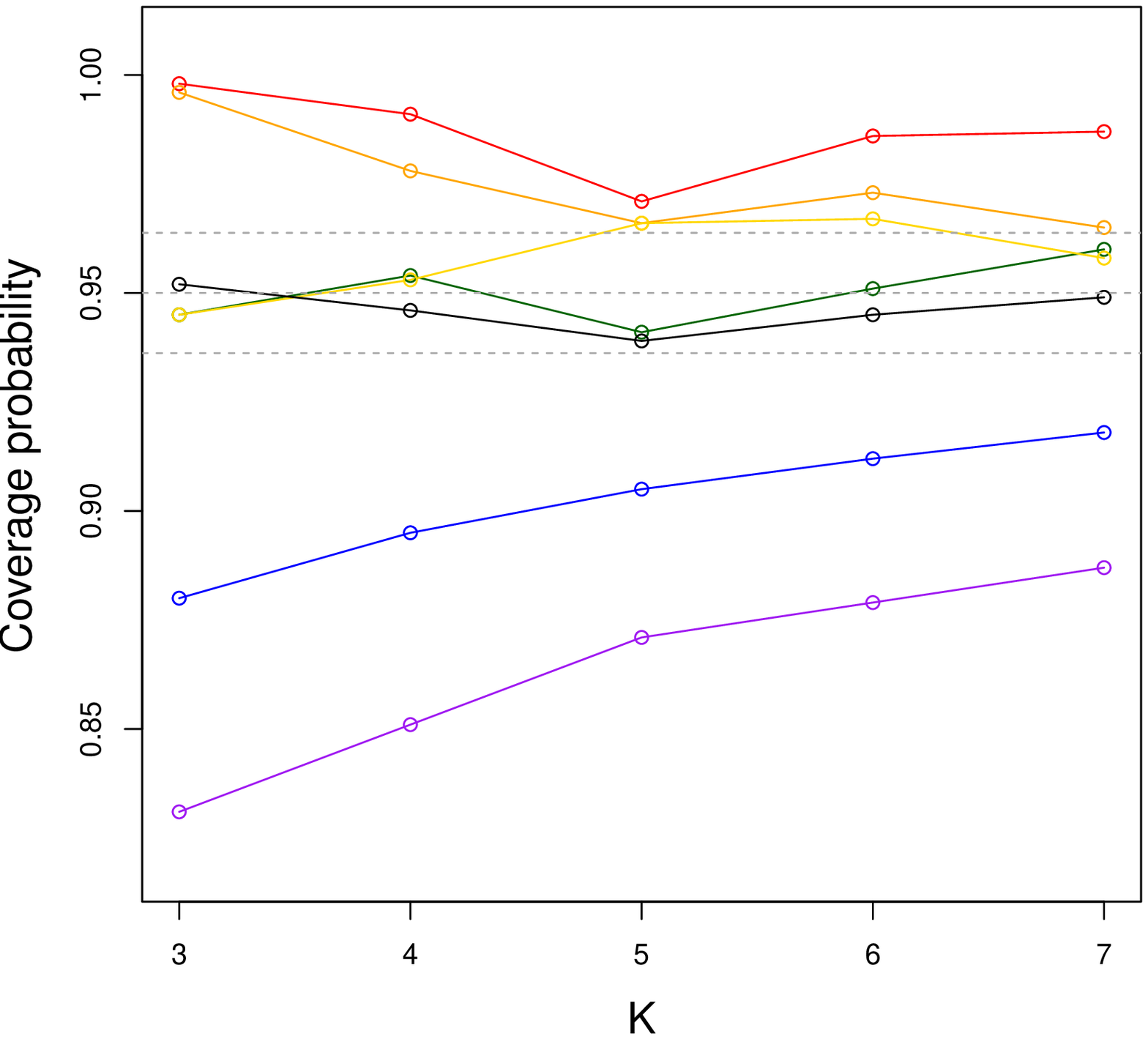}&
\includegraphics[width=.3\linewidth]{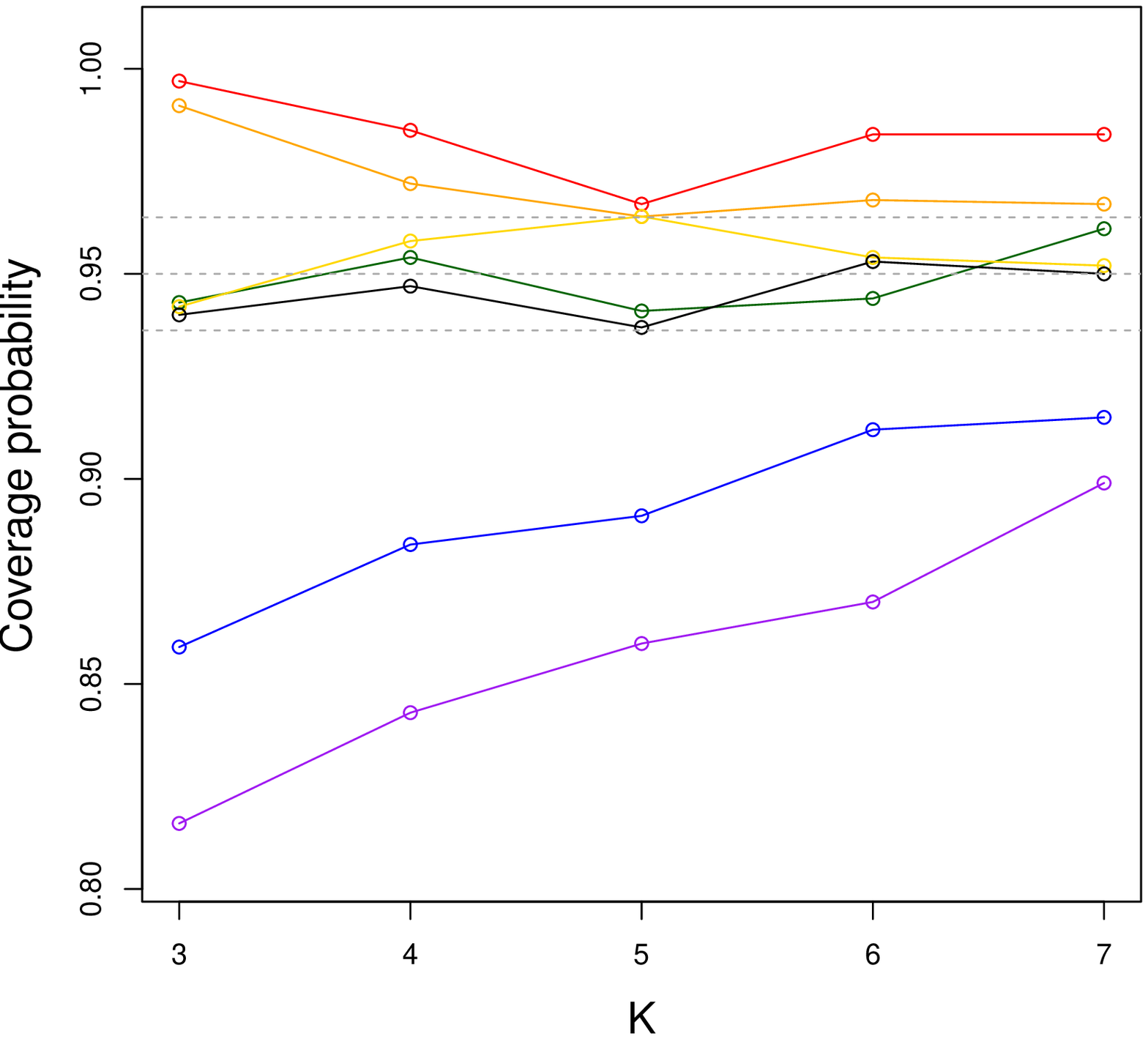}\\
\includegraphics[width=.3\linewidth]{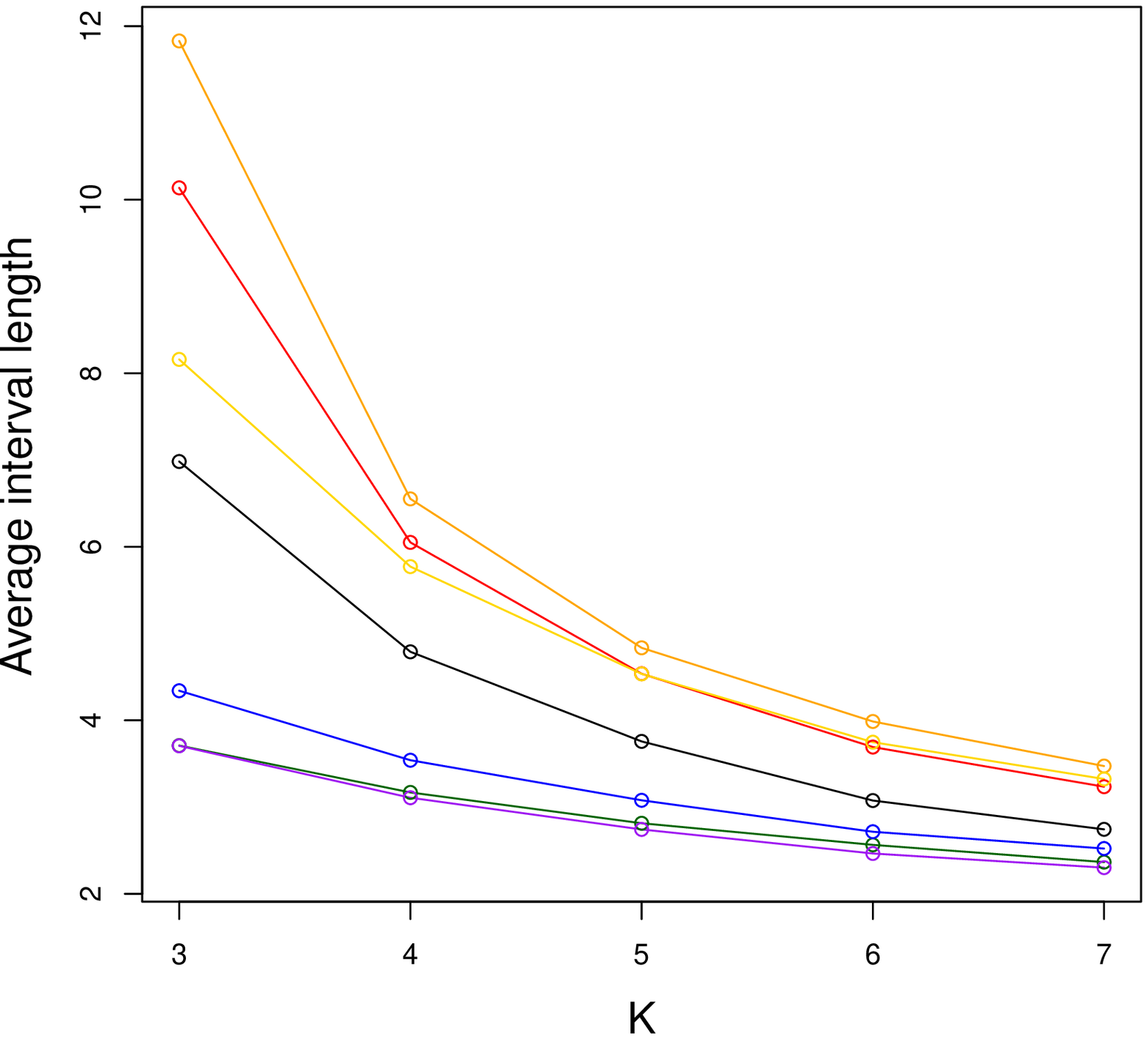}&
\includegraphics[width=.3\linewidth]{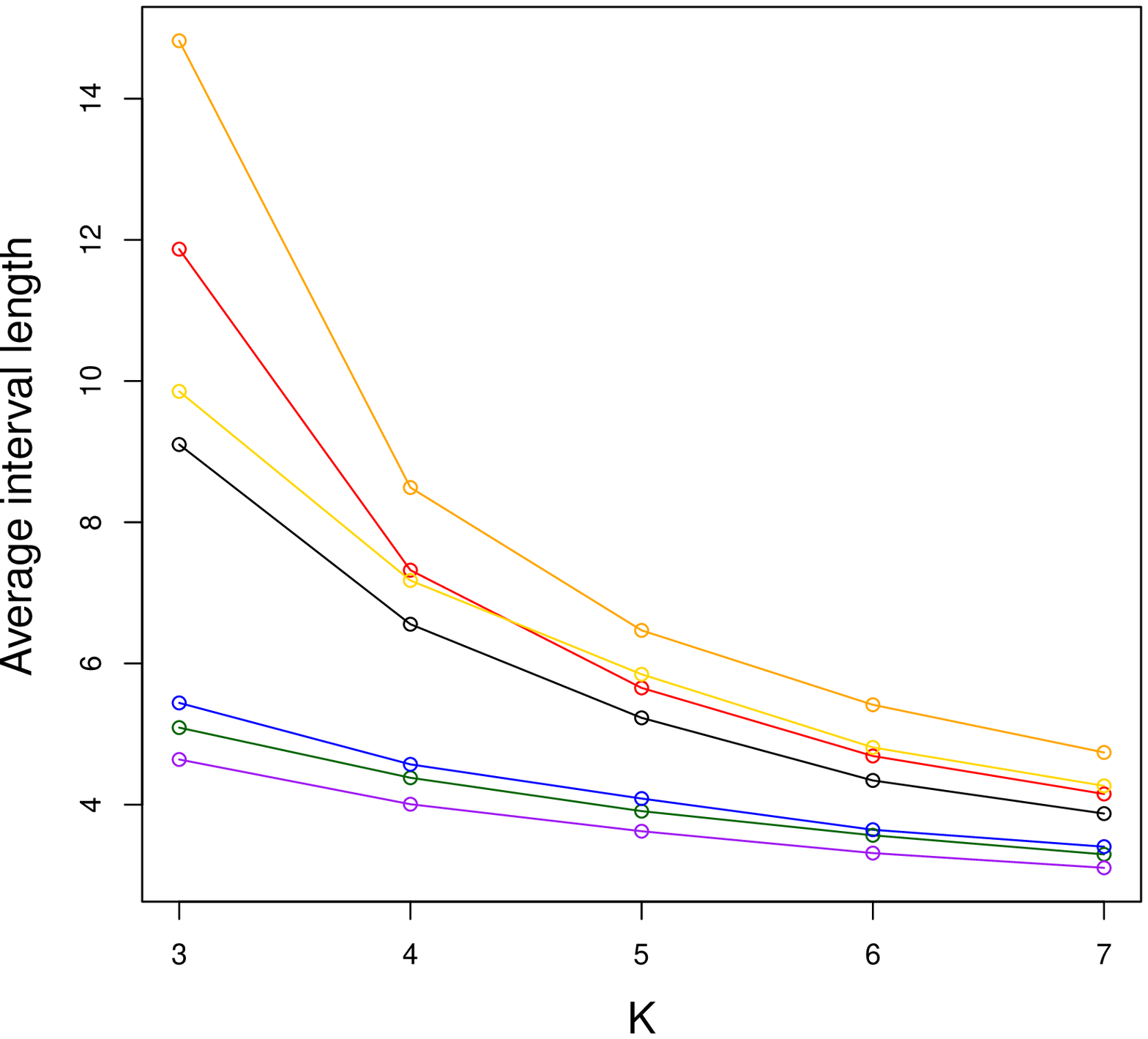}&
\includegraphics[width=.3\linewidth]{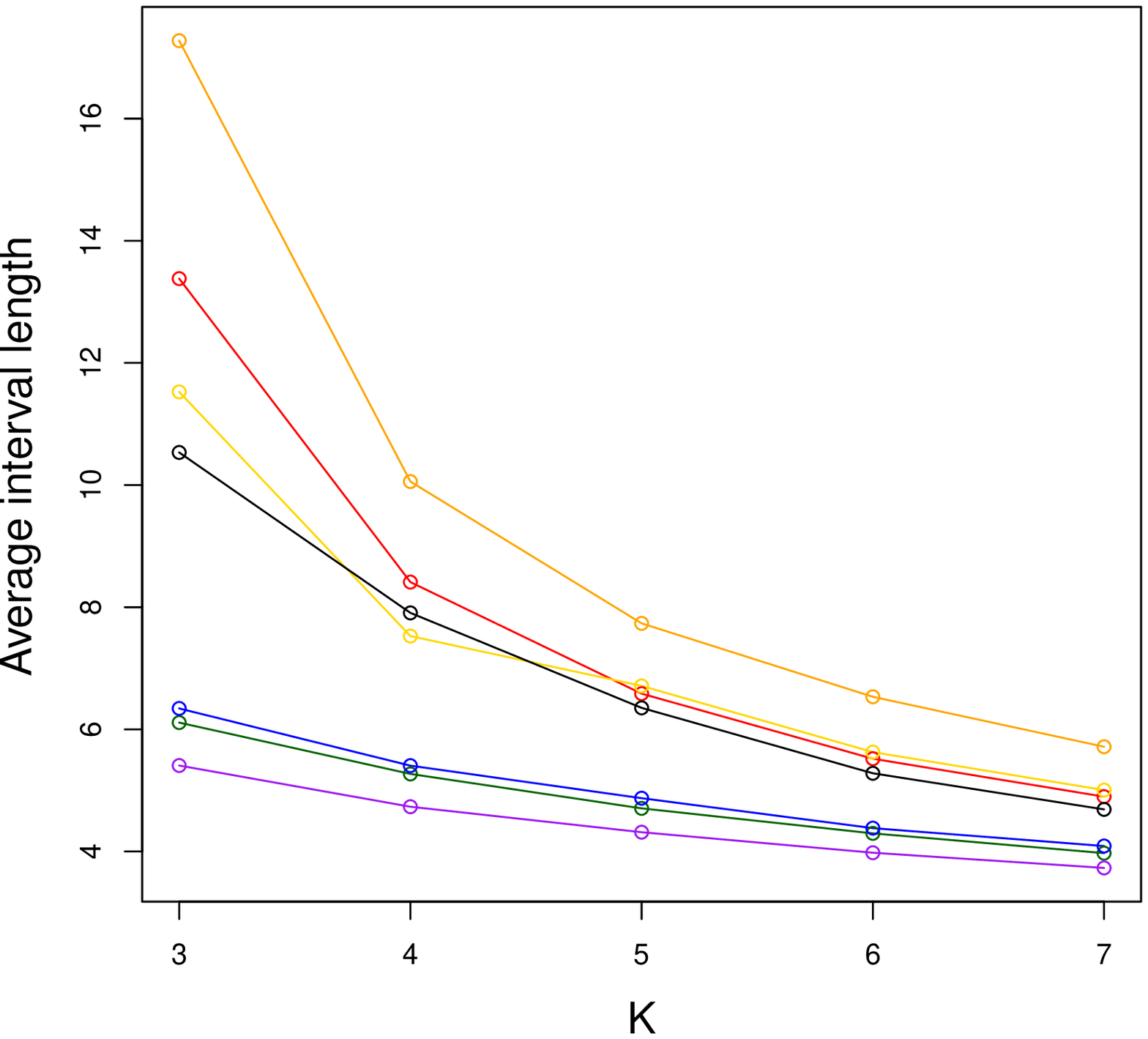}\\
\end{tabular}
\caption{Coverage probabilities and coverage lengths across 15 different simulation settings for the number of studies available for meta-analysis $K$ and the level of heterogeneity $\nu$. Results for DL (in purple), LK (in blue), oracle (in green), SV (in yellow), EX (in orange), traditional Bayes (in red), and our proposed method (in black). From left to right, data are generated from a fixed between-study variance $\nu \in \{1, 3, 5\}$.}
\label{fig:coverage_byk}
\end{figure}

\begin{figure}[t]
   \centering
\begin{tabular}{ccc}
\includegraphics[width=.3\linewidth]{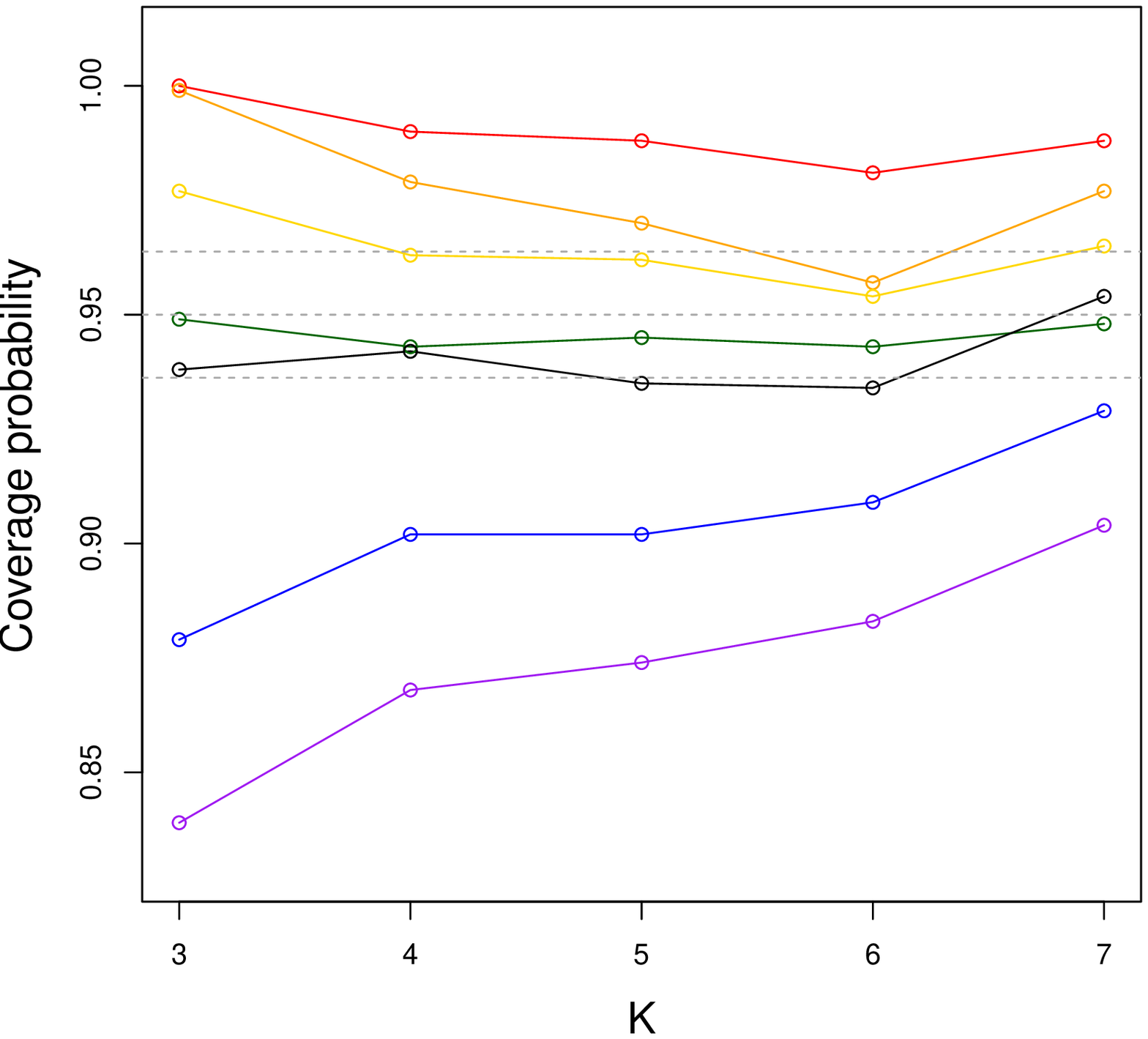}&
\includegraphics[width=.3\linewidth]{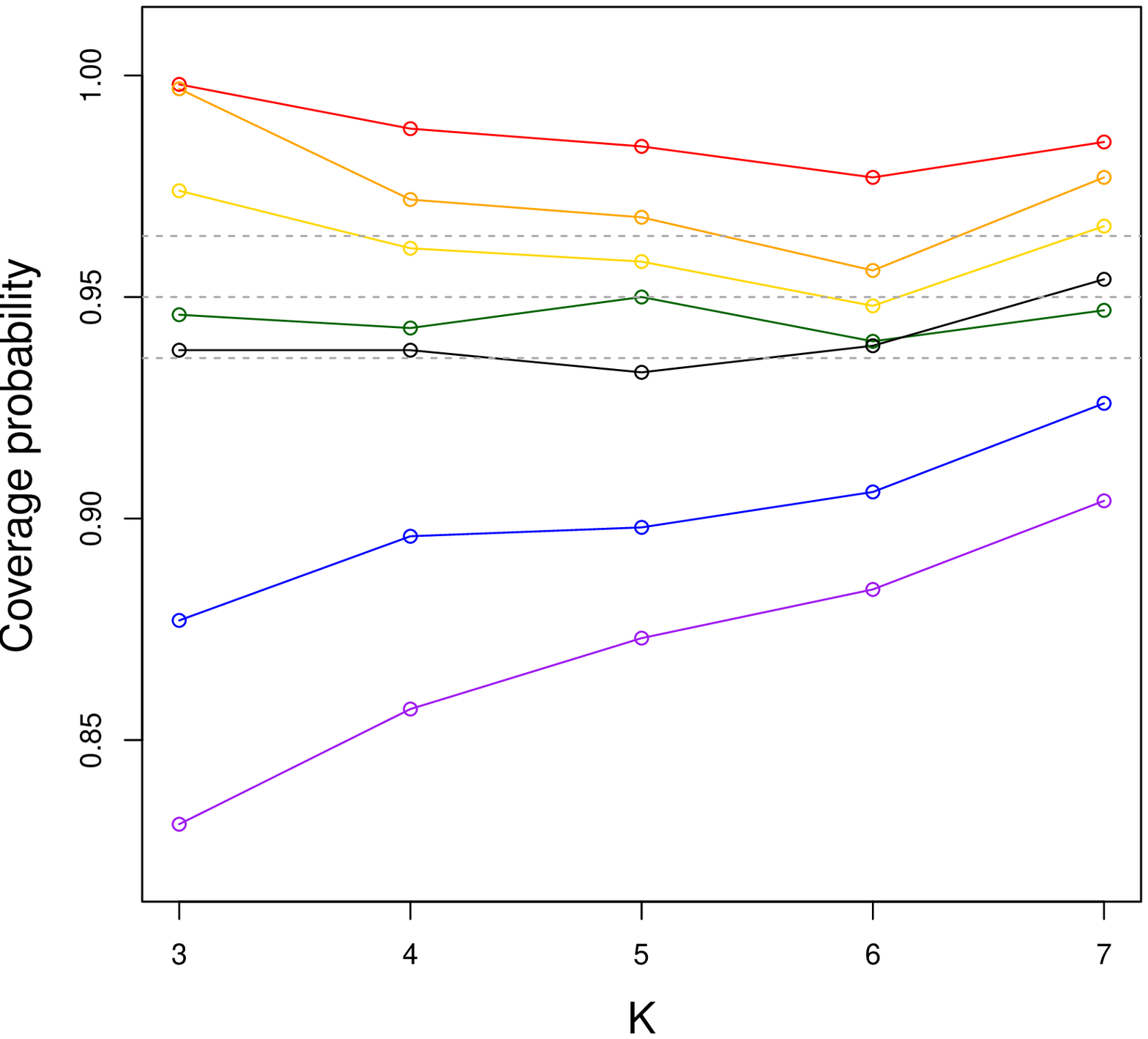}&
\includegraphics[width=.3\linewidth]{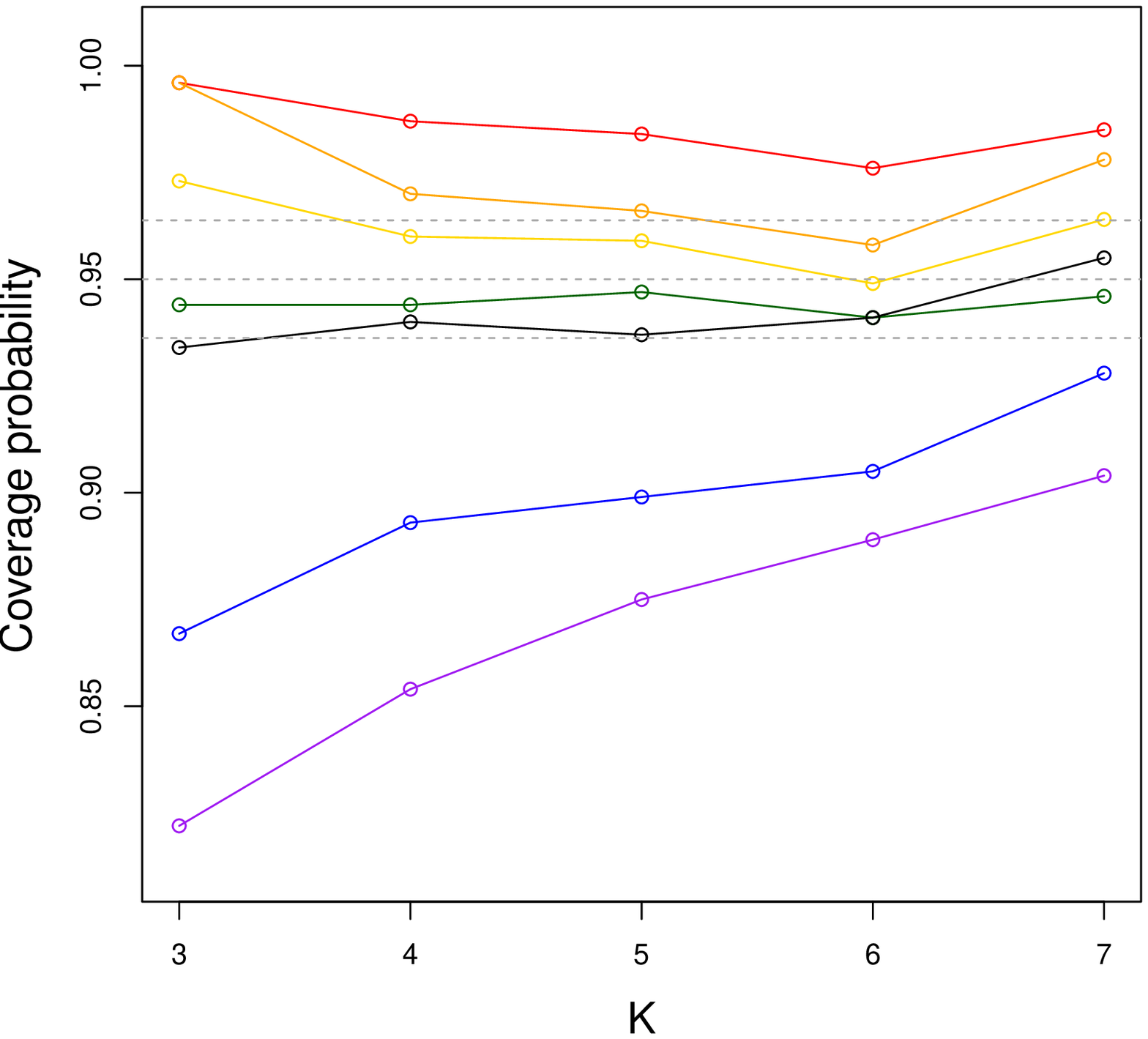}\\
\includegraphics[width=.3\linewidth]{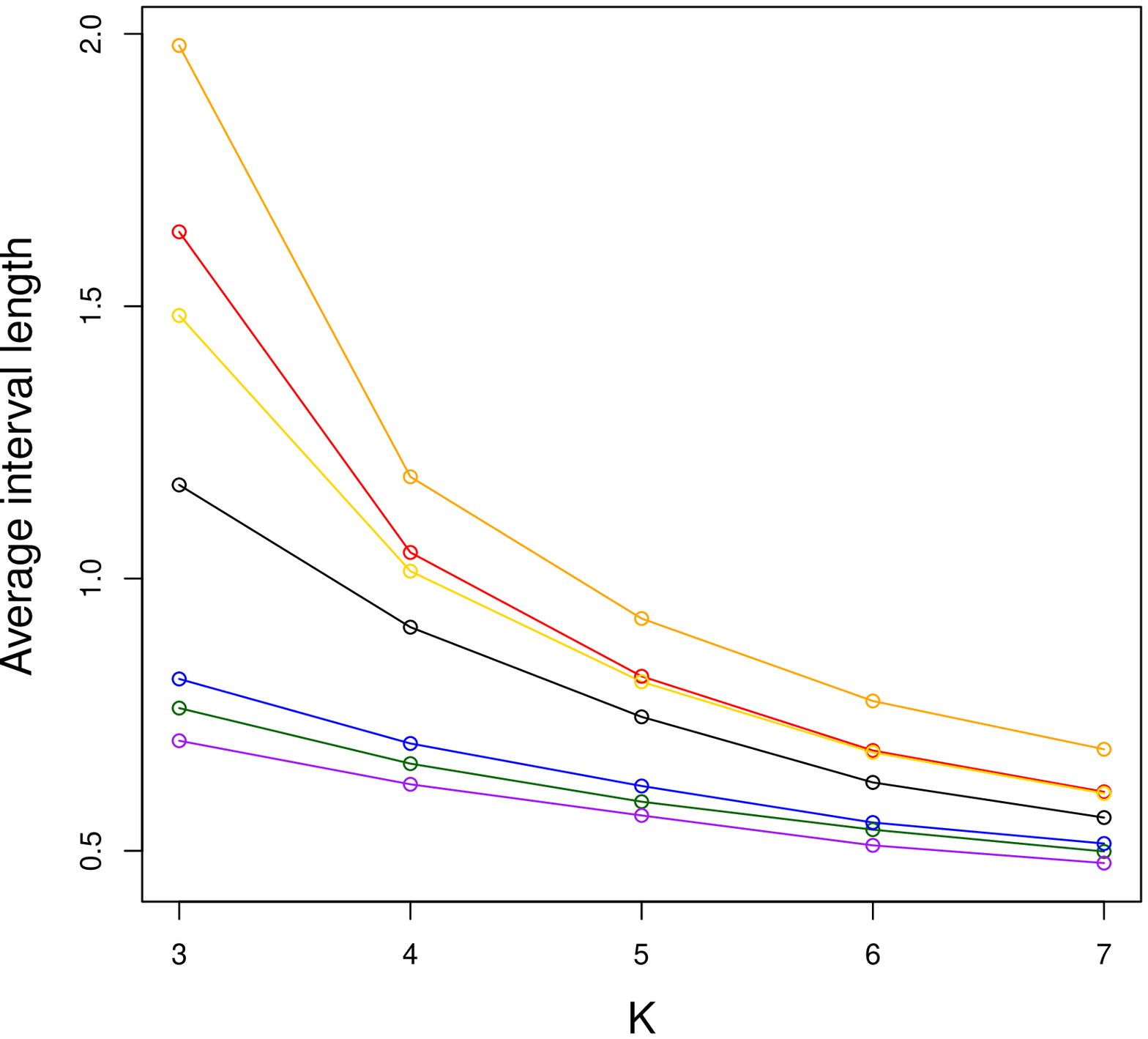}&
\includegraphics[width=.3\linewidth]{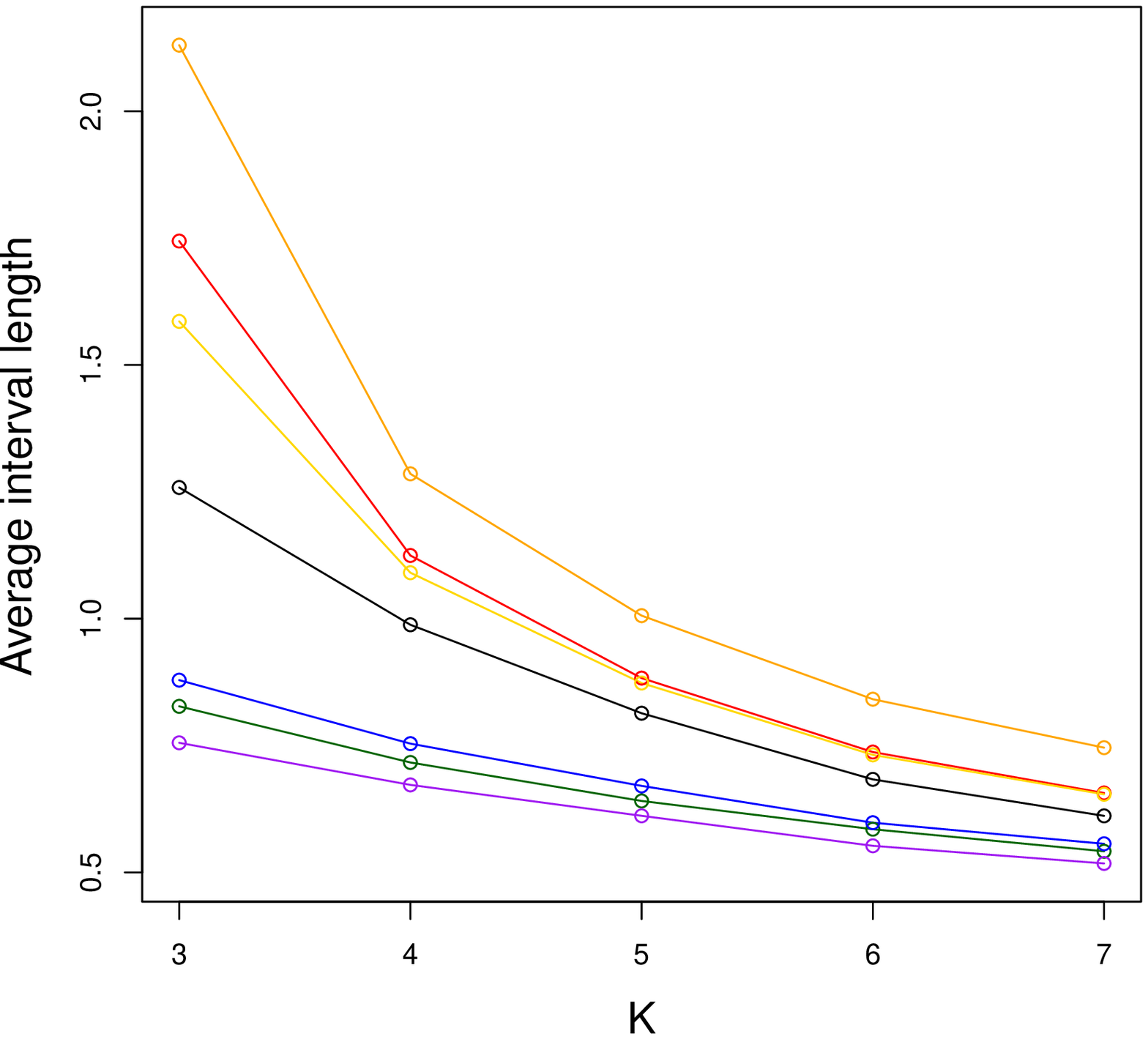}&
\includegraphics[width=.3\linewidth]{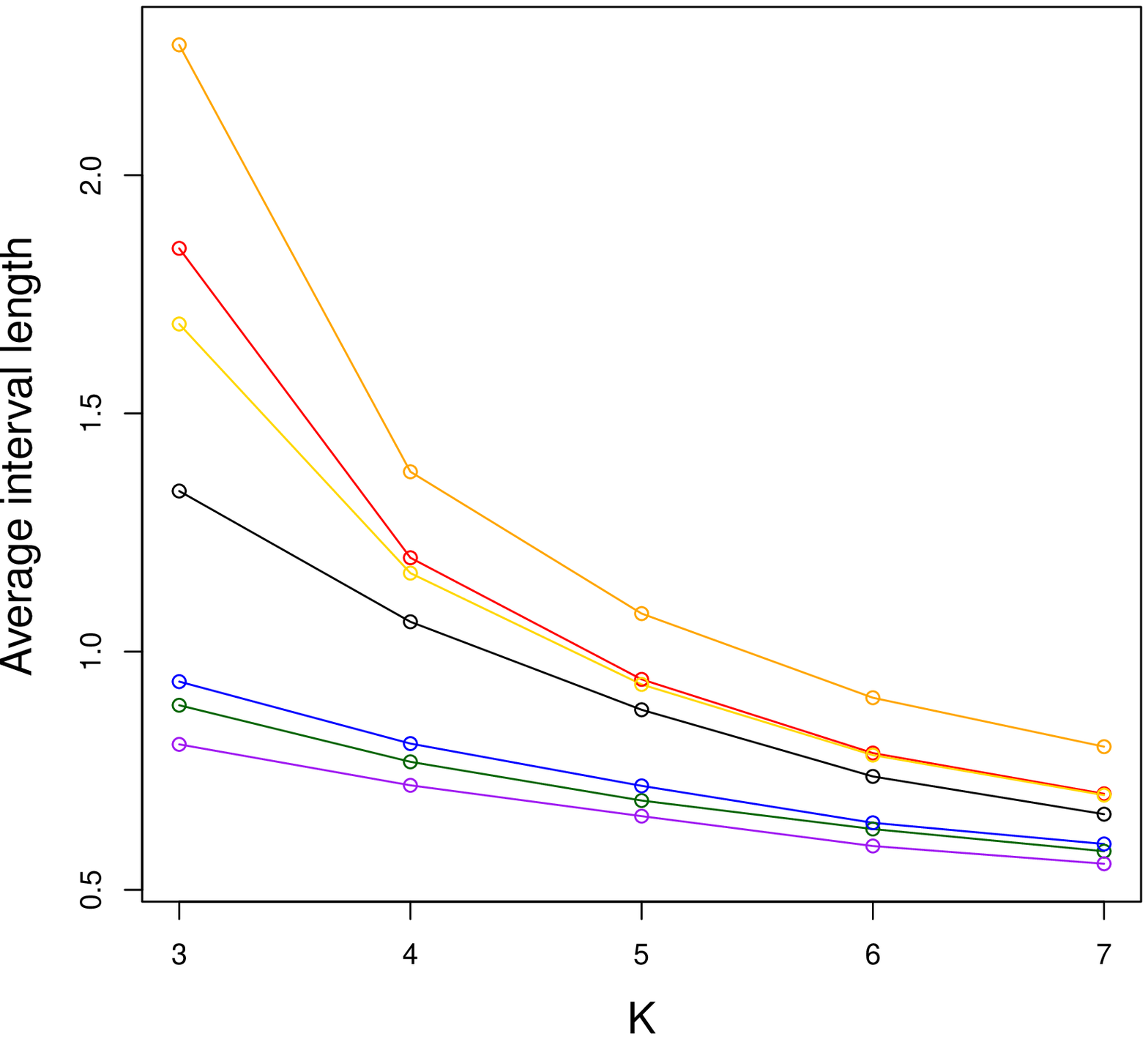}\\
\end{tabular}
\caption{Coverage probabilities and coverage lengths across 15 different simulation settings as done in \citet{guolo2012higher} for the number of studies available for meta-analysis $K$ and the level of heterogeneity $\nu$. Results for DL (in purple), LK (in blue), oracle (in green), SV (in yellow), EX (in orange), traditional Bayes (in red), and our proposed method (in black). From left to right, data are generated from a fixed between-study variance $\nu \in \{0.08, 0.10, 0.12\}$ across various number of studies included for meta-analysis.}
\label{fig:coverage_new}
\end{figure}

\section{Real data analyses}
\label{S:example}  

\subsection{Changes in bone mineral density}

To demonstrate how their inference procedure performs against popular meta-analytic techniques, \citet{michael2018} carry out four separate meta-analyses on 59 randomised trials presented in \citet{tai2015}. These meta-analyses differ in two categories: (1) the specific bones from which the outcome measure, or the change in bone mineral density (BMD), was measured; and (2) the number studies included. The first meta-analysis consisted of 27 separate trials in which BMD changes were taken from the lumbar spine, followed by a meta-analysis of six trials of BMD changes from the hip, five from the forearm, and three from the total body. As shown in Table~\ref{fig:exampletai}, for the first meta-analysis, our 95\% plausibility interval almost matches the exact confidence interval from \citet{michael2018}, which also approximately aligns with that of the classical DerSimonian--Laird approach. This is no surprise as the number of studies $K = 27$ itself is large. As for subsequent studies in which $K \leq 6$, the comparison between DerSimonial--Laird and the other two methods changes a lot and, in fact, sometimes leads to different scientific conclusions.  For example, one would conclude a significant change in BMD from the forearm and total body meta-analyses based on DerSimonian--Laird, but conclude no significant change based on our method and that of \citet{michael2018}.  Given that the latter two approaches have stronger theoretical support than the former, the difference in conclusions here might be indicative of the increased risk of false positives when using traditional meta-analytic techniques.


\begin{table}[t]
\caption{Four meta-analyses on the effect of calcium supplements in changes in bone mineral density from \citep{tai2015}. Intervals based on three methods---DerSimonian--Laird, Michael et al, and ours---are reported, with interval lengths as subscripts.}
\begin{center}
\begin{tabular}{lcccc}
\hline 
Study       & $K$  & DerSimonian--Laird                  & Michael et al               & Ours   \\ \hline 
Lumbar spine & 27 & $(0.828, 1.669)_{0.841}$ & $(0.768, 1.726)_{0.958}$  & $(0.811, 1.642)_{0.831}$  \\ 
Total hip    & 6  & $(0.502, 1.847)_{1.345}$ & $(0.159, 2.246)_{2.087}$  & $(0.319, 2.131)_{1.812}$  \\ 
Forearm      & 5  & $(0.209, 3.378)_{3.169}$ & $(-0.459, 4.124)_{4.583}$ & $(-0.426, 4.625)_{5.052}$ \\ 
Total body   & 3  & $(0.268, 1.778)_{1.511}$ & $(-0.740, 2.796)_{3.536}$ & $(-0.486, 2.568)_{3.054}$ \\ \hline
\end{tabular}
\end{center}
\label{fig:exampletai}
\end{table}

\subsection{Risk of acute myocardial infarction}
\label{S:magnesium} 

Here we consider a controversial example, one that called to question the use of meta-analyses in general \citep{egger1995, flather1997}. \citet{teo1991} conducted a meta-analysis of seven clinical trials that examined mortality across 1301 patients, 657 of which received intravenous magnesium therapy within 12 hours of hospitalization for acute myocardial infarction and 644 of which did not. In the original work, a fixed-effect method was used to combine the results from these seven randomized trials and arrive at a common odds ratio of 0.47, with 95\% confidence interval $(0.28, 0.79)$---suggesting magnesium therapy to be highly effective in reducing mortality among this specific patient population. The expected drop in mortality, however, was refuted in a large-scale 58,050-patient follow-up study \citep{isis1995} that estimated a common odds ratio of 1.06 with 95\% confidence interval $(1.00, 1.12)$. As a result, researchers raised concerns about meta-analytic techniques in general, citing issues around publication biases \citep{yusuf1995} and high heterogeneity between studies \citep{flather1997}. To address these problems, the canonical recommendation was to conduct sensitivity analyses via the use of multiple meta-analysis procedures, like that discussed below.  Had such a precaution been taken, the fact that the \citet{isis1995} study lead to an alternative conclusion would not have been unforeseen. 

It is also worth noting here that since the raw observations recorded in \citet{teo1991} are in the form of a dichotomous outcome variable, we take the logarithm of the common odds ratio, between the mortality rate of patients that receive magnesium therapy and that of patients that do not, in order to conduct our meta-analysis as described in the competing procedures above.  While there are other simplifications, e.g., \citet{van1993bivariate}, we subscribe to the rationale in \citet{dersimonian1986} that regards the distribution of the log odds as approximately normal.  Figure~\ref{fig:magnesium} thus compares the resulting interval estimates based on several meta-analytic procedures, namely, those assessed in Section~\ref{S:sims}.  Note that the DerSimonian--Laird and signed profile log likelihood intervals approximate the original results from \citet{teo1991}. However, our proposed approach, along with the full Bayesian, the higher-order likelihood, and that in \citet{michael2018}, result in an odds ratio interval that suggests magnesium therapy does not significantly affect the short-term mortality of patients with acute myocardial infarction.


\begin{figure}
\centering
\begin{minipage}{.5\textwidth}
  \centering
  \includegraphics[width=1\linewidth]{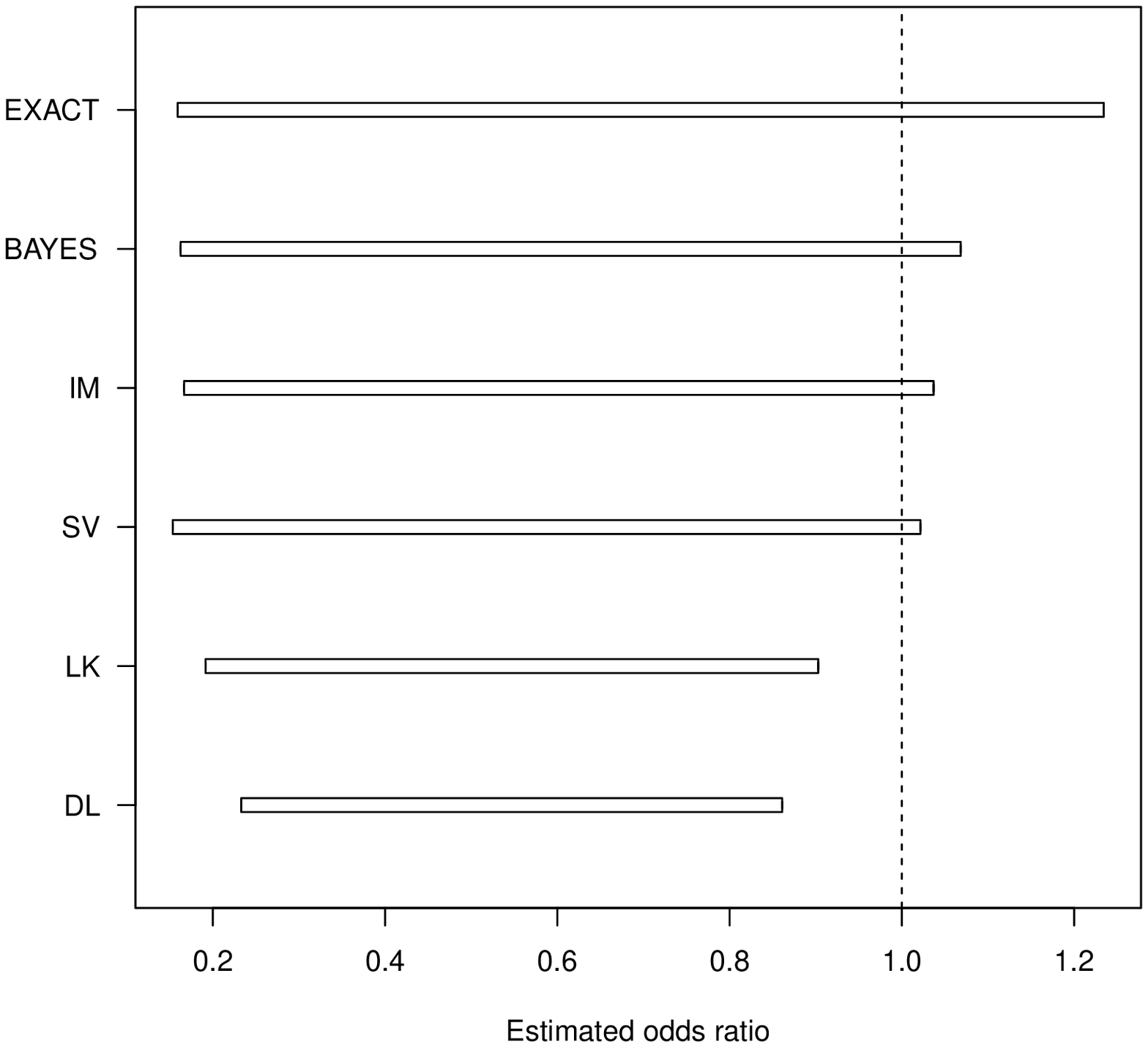}
\end{minipage}%
\begin{minipage}{.5\textwidth}
  \centering
  \includegraphics[width=1\linewidth]{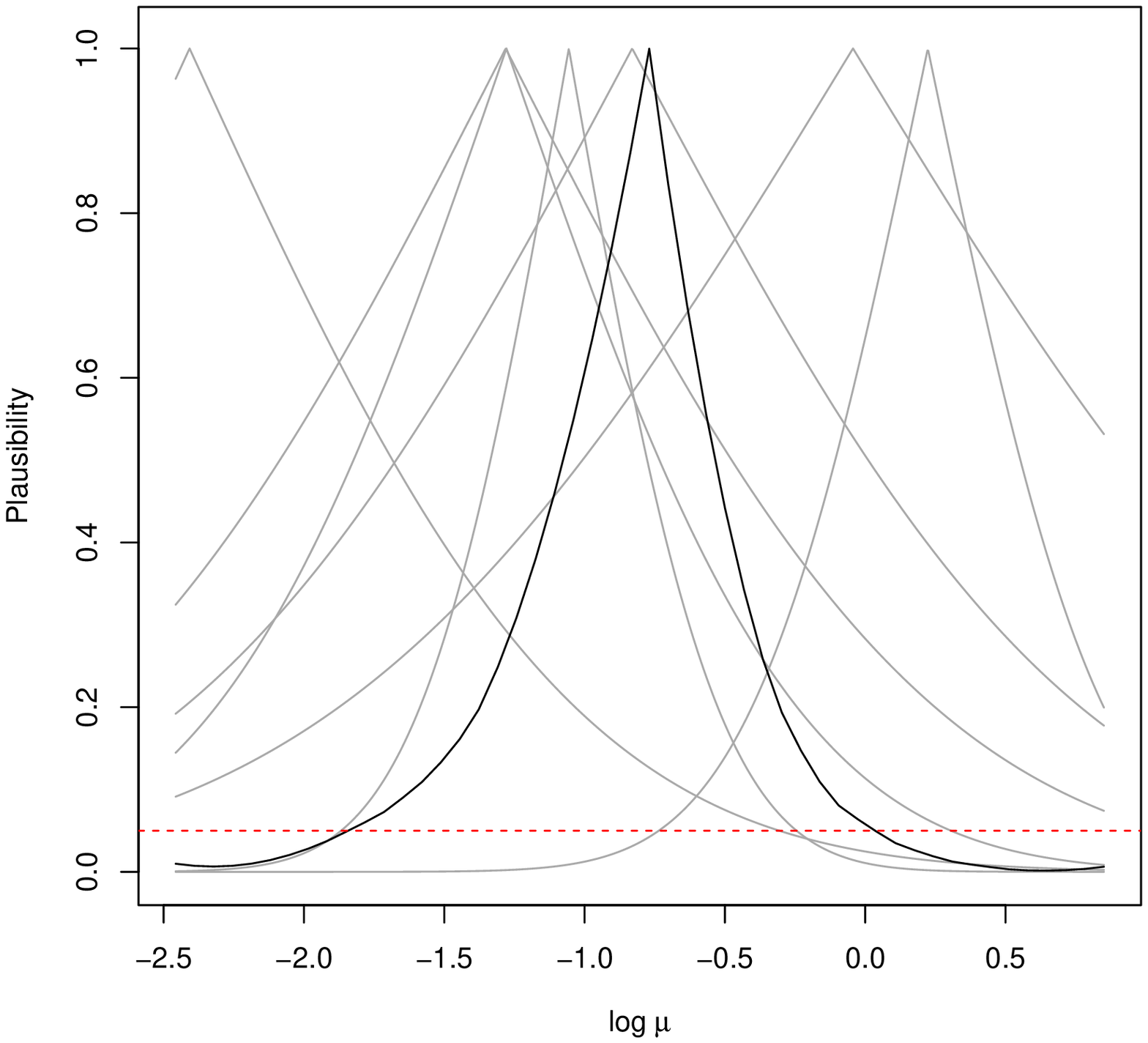}
\end{minipage}
\caption{(Left) Estimated odds ratio intervals for various meta-analytic techniques in the example described in Section~\ref{S:magnesium}. (Right) Combined plausibility function and respective IM interval on the log scale.}%
\label{fig:magnesium}
\end{figure}

\section{Conclusion}
\label{S:discuss}  

In this paper, we have considered an important and challenging problem, namely, valid statistical inference for meta-analyses that combine only a few studies.  Again, the main obstacle is in dealing with the unknown between-study variance, in which there is only limited information in the few studies being combined.  Our proposed solution is based on a recently proposed generalized inferential model framework, and we harness the power of profiling to construct a generalized association that is ``almost'' independent of the nuisance between-study variance.  From there, we can use the exact distribution of the profile likelihood ratio, as the lack of sensitivity to the nuisance parameter means that it is not necessary to have an accurate plug-in estimator to achieve near-exact inference.  In our numerical comparisons, we have demonstrated that the proposed inferential model solution outperforms existing methods in the literature, by being nearly exact and more efficient across a wide range of simulation settings, with few studies and both large and small between-study variance.  

Given the strong performance in this application, it is natural to consider using the same approach to solve other challenging problems.  One that we have recently considered is when data come from a parametric model are corrupted by random censoring.  The classical solution to this problem relies on the asymptotic normality of maximum likelihood estimators and, therefore, can only give approximately valid inference in an asymptotic sense.  But the use of a likelihood ratio effectively marginalizes out the nuisance censoring distribution, so we end up in a position similar to that encountered in the present paper, the key difference being that the nuisance parameter is infinite-dimensional, which creates computational challenges.  Preliminary results on this can be found in \citet{cahoon.martin.isipta} and more details are forthcoming.  




\appendix

\section{Proof of Theorem~\ref{thm:valid}}
\label{S:proof}

By definition of $\mpl_{Y^K}$ in \eqref{eq:mpl}, it is enough to show that  $G_{\hat\nu_\mu}^K(T_{Y^K}(\mu))$ converges in distribution to $\unif(0,1)$ under $\prob_{Y^K|\mu,\nu}$; note that here we insert the superscript ``$K$'' to highlight the dependence on the number of studies, $K$.  Then we can write $G_{\hat\nu_\mu}^K(T_{Y^K}(\mu)) = G_\nu^K(T_{Y^K}(\mu)) + \Delta_K$, where 
\[\Delta_K = G_\nu^K(T_Y(\mu)) - G_{\hat\nu_\mu}^K(T_Y(\mu)),\]
with $\hat\nu_\mu$, the maximum likelihood estimate of the heterogeneity parameter at a fixed $\mu$, and $\nu$, the true heterogeneity between studies.  The key observation is that $G_\nu^K(T_{Y^K}(\mu))$ is exactly uniformly distributed under $\prob_{Y^K|\mu, \nu}$, so if we can show $\Delta_K \to 0$ in $\prob_{Y^K|\mu, \nu}$-probability, then the claim follows from Slutsky's theorem.

Towards this, we clearly have
\[|\Delta_K| \leq \sup_{t\in [0, 1]} |G_\nu^K(t) - G_{\hat\nu_\mu}^K(t)|,\]
so we can prove our claim by showing the difference between the two distribution functions vanishes uniformly.  But since these are distribution functions, it is enough to show that the difference vanishes pointwise, at each fixed $t$.  Towards this, according to \citet{guolo2012higher}, the meta-analysis problem is sufficiently regular that the classical first-order distribution theory applies; see, e.g., \citet[][Sec.~4.6]{severini1992profile}.  In particular, this implies $\hat\nu_\mu = \nu + O_P(K^{-1/2})$ which, in turn, implies that $\prob_{Y^K|\mu,\nu}$ and $\prob_{Y^K|\mu, \hat\nu_\mu}$ are mutually contiguous.  Then the classical Wilks's theorem gives us 
\begin{equation}
-2\log T_{Y^K}(\mu) \to \chisq(1) \quad \text{in distribution, as $K \to \infty$},
\label{eq:chisq}
\end{equation}
under {\em both} $\prob_{Y^K|\mu, \nu}$ and $\prob_{Y^K|\mu, \hat\nu_\mu}$.  
Therefore, 
\[G_\nu^K(t) \to G^\infty(t) \quad \text{and} \quad G_{\hat\nu_\mu}^K(t) \to G^\infty(t), \]
where $G^\infty$ is the limiting distribution function of $T_{Y^K}(\mu)$ from \eqref{eq:chisq}.  If we then express
\[ |G_\nu^K(t) - G_{\hat\nu_\mu}^K(t)| \leq |G_\nu^K(t) - G^\infty (t)| + |G_{\hat\nu_\mu}^K(t) - G^\infty (t)|,\]
we see the right-hand converges to $0$ in $\prob_{Y|\mu, \nu}$-probability as $K \to \infty$.  This implies $\Delta_K \to 0$ from which $G_{\hat\nu_\mu}^K(T_{Y^K}(\mu)) \to \unif(0,1)$ follows by Slutsky's theorem.

\bibliographystyle{apa}
\bibliography{ref.bib}

\begin{thebibliography}{}

\bibitem[\protect\astroncite{Balch et~al.}{2019}]{balch2019}
Balch, M.~S., Martin, R., and Ferson, S. (2019).
\newblock Satellite conjunction analysis and the false confidence theorem.
\newblock {\em Proceedings of the Royal Society A}, 475(2227):20180565.

\bibitem[\protect\astroncite{Barnard}{1995}]{barnard1995pivotal}
Barnard, G.~A. (1995).
\newblock Pivotal models and the fiducial argument.
\newblock {\em International Statistical Review/Revue Internationale de
  Statistique}, pages 309--323.

\bibitem[\protect\astroncite{Cahoon and Martin}{2019}]{cahoon.martin.isipta}
Cahoon, J. and Martin, R. (2019).
\newblock Possibility measures for valid statistical inference based on
  censored data.
\newblock In De~Bock, J., {de Campos}, C.~P., {de Cooman}, G., Quaeghebeur, E.,
  and Wheeler, G., editors, {\em Proceedings of the Eleventh International
  Symposium on Imprecise Probabilities: Theories and Applications}, volume 103
  of {\em Proceedings of Machine Learning Research}, pages 49--58, Thagaste,
  Ghent, Belgium. PMLR.

\bibitem[\protect\astroncite{Chung et~al.}{2013}]{chung2013}
Chung, Y., Rabe-Hesketh, S., and Choi, I.-H. (2013).
\newblock Avoiding zero between-study variance estimates in random-effects
  meta-analysis.
\newblock {\em Statistics in Medicine}, 32(23):4071--4089.

\bibitem[\protect\astroncite{Cochran}{1954}]{cochran1954}
Cochran, W.~G. (1954).
\newblock The combination of estimates from different experiments.
\newblock {\em Biometrics}, 10(1):101--129.

\bibitem[\protect\astroncite{Davey et~al.}{2011}]{davey2011}
Davey, J., Turner, R.~M., Clarke, M.~J., and Higgins, J.~P. (2011).
\newblock Characteristics of meta-analyses and their component studies in the
  {C}ochrane database of systematic reviews: a cross-sectional, descriptive
  analysis.
\newblock {\em BMC Medical Research Methodology}, 11(1):160.

\bibitem[\protect\astroncite{Dawid and Stone}{1982}]{dawid1982functional}
Dawid, A.~P. and Stone, M. (1982).
\newblock The functional-model basis of fiducial inference.
\newblock {\em The Annals of Statistics}, 10(4):1054--1067.

\bibitem[\protect\astroncite{Demidenko}{2013}]{demidenko2013}
Demidenko, E. (2013).
\newblock {\em Mixed {M}odels: {T}heory and {A}pplications with {R}}.
\newblock John Wiley \& Sons.

\bibitem[\protect\astroncite{Dempster}{2008}]{dempster2008}
Dempster, A.~P. (2008).
\newblock The {D}empster--{S}hafer calculus for statisticians.
\newblock {\em International Journal of Approximate Reasoning}, 48(2):365--377.

\bibitem[\protect\astroncite{DerSimonian and Kacker}{2007}]{dersimonian2007}
DerSimonian, R. and Kacker, R. (2007).
\newblock Random-effects model for meta-analysis of clinical trials: an update.
\newblock {\em Contemporary Clinical Trials}, 28(2):105--114.

\bibitem[\protect\astroncite{DerSimonian and Laird}{1986}]{dersimonian1986}
DerSimonian, R. and Laird, N. (1986).
\newblock Meta-analysis in clinical trials.
\newblock {\em Controlled Clinical Trials}, 7(3):177--188.

\bibitem[\protect\astroncite{Egger and Smith}{1995}]{egger1995}
Egger, M. and Smith, G.~D. (1995).
\newblock Misleading meta-analysis.
\newblock {\em British Medical Journal Publishing Group}.

\bibitem[\protect\astroncite{Fisher}{1956}]{fisher1956}
Fisher, R.~A. (1956).
\newblock {\em Statistical {M}ethods and {S}cientific {I}nference.}
\newblock Hafner Publishing Co.

\bibitem[\protect\astroncite{Flather et~al.}{1997}]{flather1997}
Flather, M.~D., Farkouh, M.~E., Pogue, J.~M., and Yusuf, S. (1997).
\newblock Strengths and limitations of meta-analysis: larger studies may be
  more reliable.
\newblock {\em Controlled Clinical Trials}, 18(6):568--579.

\bibitem[\protect\astroncite{Follmann and Proschan}{1999}]{follmann1999}
Follmann, D.~A. and Proschan, M.~A. (1999).
\newblock Valid inference in random effects meta-analysis.
\newblock {\em Biometrics}, 55(3):732--737.

\bibitem[\protect\astroncite{{F}ourth International Infarct
  Survival~Collaborative}{1995}]{isis1995}
{F}ourth International Infarct Survival~Collaborative (1995).
\newblock {ISIS--4}: a randomised factorial trial assessing early oral
  captopril, oral mononitrate, and intravenous magnesium sulphate in 58,050
  patients with suspected acute myocardial infarction.
\newblock {\em Lancet}, 345(8951):669--685.

\bibitem[\protect\astroncite{Fraser}{1968}]{fraser1968}
Fraser, D. (1968).
\newblock {\em The {S}tructure of {I}nference.}
\newblock Wiley.

\bibitem[\protect\astroncite{Gelman}{2006}]{gelman2006}
Gelman, A. (2006).
\newblock Prior distributions for variance parameters in hierarchical models.
\newblock {\em Bayesian Analysis}, 1(3):515--534.

\bibitem[\protect\astroncite{Goodman}{1989}]{goodman1989meta}
Goodman, S.~N. (1989).
\newblock Meta-analysis and evidence.
\newblock {\em Controlled Clinical Trials}, 10(2):188--204.

\bibitem[\protect\astroncite{Guolo}{2012}]{guolo2012higher}
Guolo, A. (2012).
\newblock Higher-order likelihood inference in meta-analysis and
  meta-regression.
\newblock {\em Statistics in Medicine}, 31(4):313--327.

\bibitem[\protect\astroncite{Guolo and Varin}{2012}]{guolo2012r}
Guolo, A. and Varin, C. (2012).
\newblock The {R} package \texttt{metaLik} for likelihood inference in
  meta-analysis.
\newblock {\em Journal of Statistical Software}, 50(7):1--14.

\bibitem[\protect\astroncite{Hardy and Thompson}{1996}]{hardy1996}
Hardy, R.~J. and Thompson, S.~G. (1996).
\newblock A likelihood approach to meta-analysis with random effects.
\newblock {\em Statistics in Medicine}, 15(6):619--629.

\bibitem[\protect\astroncite{Jackson et~al.}{2010}]{jackson2010}
Jackson, D., Bowden, J., and Baker, R. (2010).
\newblock How does the {D}er{S}imonian and {L}aird procedure for random effects
  meta-analysis compare with its more efficient but harder to compute
  counterparts?
\newblock {\em Journal of Statistical Planning and Inference}, 140(4):961--970.

\bibitem[\protect\astroncite{Kosorok}{2008}]{kosorok2008}
Kosorok, M.~R. (2008).
\newblock {\em Introduction to {E}mpirical {P}rocesses and {S}emiparametric
  {I}nference.}
\newblock Springer.

\bibitem[\protect\astroncite{Liu et~al.}{2018}]{liu2018}
Liu, S., Tian, L., Lee, S., and Xie, M.-g. (2018).
\newblock Exact inference on meta-analysis with generalized fixed-effects and
  random-effects models.
\newblock {\em Biostatistics \& Epidemiology}, 2(1):1--22.

\bibitem[\protect\astroncite{Martin}{2015}]{martin2015}
Martin, R. (2015).
\newblock Plausibility functions and exact frequentist inference.
\newblock {\em Journal of the American Statistical Association},
  110(512):1552--1561.

\bibitem[\protect\astroncite{Martin}{2017}]{pvalue.course}
Martin, R. (2017).
\newblock A statistical inference course based on {$p$}-values.
\newblock {\em The American Statistician}, 71(2):128--136.

\bibitem[\protect\astroncite{Martin}{2018}]{martin2018}
Martin, R. (2018).
\newblock On an inferential model construction using generalized associations.
\newblock {\em Journal of Statistical Planning and Inference}, 195:105--115.

\bibitem[\protect\astroncite{Martin}{2019}]{martin2019false}
Martin, R. (2019).
\newblock False confidence, non-additive beliefs, and valid statistical
  inference.
\newblock {\em International Journal of Approximate Reasoning}, 113:39--73.

\bibitem[\protect\astroncite{Martin and Liu}{2013}]{martin2013}
Martin, R. and Liu, C. (2013).
\newblock Inferential models: a framework for prior-free posterior
  probabilistic inference.
\newblock {\em Journal of the American Statistical Association},
  108(501):301--313.

\bibitem[\protect\astroncite{Martin and Liu}{2015a}]{martin2015conditional}
Martin, R. and Liu, C. (2015a).
\newblock Conditional inferential models: combining information for prior-free
  probabilistic inference.
\newblock {\em Journal of the Royal Statistical Society: Series B (Statistical
  Methodology)}, 77(1):195--217.

\bibitem[\protect\astroncite{Martin and Liu}{2015b}]{imbook}
Martin, R. and Liu, C. (2015b).
\newblock {\em Inferential {M}odels: {R}easoning with {U}ncertainty}.
\newblock CRC Press, Boca Raton, FL.

\bibitem[\protect\astroncite{Martin and Liu}{2015c}]{martin2015marginal}
Martin, R. and Liu, C. (2015c).
\newblock Marginal inferential models: prior-free probabilistic inference on
  interest parameters.
\newblock {\em Journal of the American Statistical Association},
  110(512):1621--1631.

\bibitem[\protect\astroncite{Michael et~al.}{2019}]{michael2018}
Michael, H., Thornton, S., Xie, M., and Tian, L. (2019).
\newblock Exact inference on the random-effects model for meta-analyses with
  few studies.
\newblock {\em Biometrics}, 75(2):485--493.

\bibitem[\protect\astroncite{Paule and Mandel}{1982}]{paule1982}
Paule, R.~C. and Mandel, J. (1982).
\newblock Consensus values and weighting factors.
\newblock {\em Journal of Research of the National Bureau of Standards},
  87(5):377--385.

\bibitem[\protect\astroncite{R{\"o}ver}{2017}]{rover2017}
R{\"o}ver, C. (2017).
\newblock Bayesian random-effects meta-analysis using the \texttt{bayesmeta}
  {R} package.
\newblock \url{https://arxiv.org/abs/1711.08683/}.

\bibitem[\protect\astroncite{Severini}{2000}]{severini2000likelihood}
Severini, T.~A. (2000).
\newblock {\em Likelihood {M}ethods in {S}tatistics}.
\newblock Oxford University Press.

\bibitem[\protect\astroncite{Shafer}{1976}]{shafer1976}
Shafer, G. (1976).
\newblock {\em A {M}athematical {T}heory of {E}vidence}.
\newblock Princeton University Press.

\bibitem[\protect\astroncite{Shafer}{1987}]{shafer1987}
Shafer, G. (1987).
\newblock Belief functions and possibility measures.
\newblock In Bezdek, J., editor, {\em The Analysis of Fuzzy Information, Vol.
  1: Mathematics and Logic}, pages 51--84. CRC Press.

\bibitem[\protect\astroncite{Sidik and Jonkman}{2007}]{sidik2007}
Sidik, K. and Jonkman, J.~N. (2007).
\newblock A comparison of heterogeneity variance estimators in combining
  results of studies.
\newblock {\em Statistics in Medicine}, 26(9):1964--1981.

\bibitem[\protect\astroncite{Tai et~al.}{2015}]{tai2015}
Tai, V., Leung, W., Grey, A., Reid, I.~R., and Bolland, M.~J. (2015).
\newblock Calcium intake and bone mineral density: systematic review and
  meta-analysis.
\newblock {\em British Medical Journal}, 351:41--83.

\bibitem[\protect\astroncite{Taraldsen et~al.}{2013}]{taraldsen2013fiducial}
Taraldsen, G., Lindqvist, B.~H., et~al. (2013).
\newblock Fiducial theory and optimal inference.
\newblock {\em The Annals of Statistics}, 41(1):323--341.

\bibitem[\protect\astroncite{Teo et~al.}{1991}]{teo1991}
Teo, K.~K., Yusuf, S., Collins, R., Held, P.~H., and Peto, R. (1991).
\newblock Effects of intravenous magnesium in suspected acute myocardial
  infarction: Overview of randomised trials.
\newblock {\em British Medical Journal}, 303(6816):1499--1503.

\bibitem[\protect\astroncite{Van~Houwelingen et~al.}{1993}]{van1993bivariate}
Van~Houwelingen, H.~C., Zwinderman, K.~H., and Stijnen, T. (1993).
\newblock A bivariate approach to meta-analysis.
\newblock {\em Statistics in medicine}, 12(24):2273--2284.

\bibitem[\protect\astroncite{Veroniki et~al.}{2016}]{veroniki2016}
Veroniki, A.~A., Jackson, D., Viechtbauer, W., Bender, R., Bowden, J., Knapp,
  G., Kuss, O., Higgins, J.~P., Langan, D., and Salanti, G. (2016).
\newblock Methods to estimate the between-study variance and its uncertainty in
  meta-analysis.
\newblock {\em Research Synthesis Methods}, 7(1):55--79.

\bibitem[\protect\astroncite{Viechtbauer}{2005}]{viechtbauer2005}
Viechtbauer, W. (2005).
\newblock Bias and efficiency of meta-analytic variance estimators in the
  random-effects model.
\newblock {\em Journal of Educational and Behavioral Statistics},
  30(3):261--293.

\bibitem[\protect\astroncite{Wang and Tian}{2018}]{wang2018}
Wang, Y. and Tian, L. (2018).
\newblock An efficient numerical algorithm for exact inference in meta
  analysis.
\newblock {\em Journal of Statistical Computation and Simulation},
  88(4):646--656.

\bibitem[\protect\astroncite{Yusuf and Flather}{1995}]{yusuf1995}
Yusuf, S. and Flather, M. (1995).
\newblock Magnesium in acute myocardial infarction.
\newblock {\em British Medical Journal}, 310:751--752.

\end{thebibliography}

\end{document}